\documentclass[11pt,letterpaper]{article}

\def\showauthornotes{0}
\def\showtableofcontents{0}
\def\showkeys{0}
\def\showdraftbox{0}
\def\showcolorlinks{1}
\def\usemicrotype{1}
\def\showfixme{1}

\def\writemode{0}

\usepackage{etex}

\usepackage[l2tabu, orthodox]{nag}

\usepackage{xspace,enumerate}

\usepackage[dvipsnames]{xcolor}

\usepackage[T1]{fontenc}
\usepackage[full]{textcomp}

\usepackage[american]{babel}

\usepackage{mathtools}

\usepackage{amsthm}

\newtheorem{theorem}{Theorem}[section]
\newtheorem*{theorem*}{Theorem}

\newtheorem*{proposition*}{Proposition}
\newtheorem{lemma}[theorem]{Lemma}
\newtheorem*{lemma*}{Lemma}
\newtheorem{corollary}[theorem]{Corollary}
\newtheorem*{conjecture*}{Conjecture}
\newtheorem{fact}[theorem]{Fact}
\newtheorem*{fact*}{Fact}

\newtheorem*{hypothesis*}{Hypothesis}

\theoremstyle{definition}
\newtheorem{definition}[theorem]{Definition}

\theoremstyle{remark}
\newtheorem{claim}[theorem]{Claim}
\newtheorem*{claim*}{Claim}
\newtheorem{remark}[theorem]{Remark}
\newtheorem*{remark*}{Remark}

\newtheorem*{observation*}{Observation}

\ifnum\writemode=1
\usepackage[
letterpaper,
top=1.2in,
bottom=1.2in,
left=1in,
right=1in]{geometry}

\fi

\ifnum\writemode=0
\usepackage[
letterpaper,
top=0.7in,
bottom=0.9in,
left=1in,
right=1in]{geometry}
\fi

\usepackage{newpxtext} 
\usepackage{textcomp} 
\usepackage[varg,bigdelims]{newpxmath}
\usepackage[scr=rsfso]{mathalfa}
\usepackage{bm} 
\let\mathbb\varmathbb

\ifnum\showkeys=1
\usepackage[color]{showkeys}
\fi

\ifnum\showcolorlinks=1
\usepackage[
pagebackref,
colorlinks=true,
urlcolor=blue,
linkcolor=blue,
citecolor=OliveGreen,
]{hyperref}
\fi

\ifnum\showcolorlinks=0
\usepackage[
pagebackref,
colorlinks=false,
pdfborder={0 0 0}
]{hyperref}
\fi

\usepackage[capitalise,nameinlink]{cleveref}
\crefname{lemma}{Lemma}{Lemmas}
\crefname{definition}{Definition}{Definitions}

\newcommand{\Sref}[1]{\hyperref[#1]{\S\ref*{#1}}}

\usepackage{nicefrac}

\let\nfrac=\nicefrac

\ifnum\usemicrotype=1
\usepackage{microtype}
\fi

\ifnum\showauthornotes=1
\newcommand{\Authornote}[2]{{\sffamily\small\color{red}{[#1: #2]}}}
\newcommand{\Authornotecolored}[3]{{\sffamily\small\color{#1}{[#2: #3]}}}
\newcommand{\Authorcomment}[2]{{\sffamily\small\color{gray}{[#1: #2]}}}
\newcommand{\Authorstartcomment}[1]{\sffamily\small\color{gray}[#1: }

\newcommand{\Authorfnote}[2]{\footnote{\color{red}{#1: #2}}}
\newcommand{\Authorfixme}[1]{\Authornote{#1}{\textbf{??}}}
\newcommand{\Authormarginmark}[1]{\marginpar{\textcolor{red}{\fbox{\Large #1:!}}}}
\else
\newcommand{\Authornote}[2]{}
\newcommand{\Authornotecolored}[3]{}
\newcommand{\Authorcomment}[2]{}
\newcommand{\Authorstartcomment}[1]{}

\newcommand{\Authorfnote}[2]{}
\newcommand{\Authorfixme}[1]{}
\newcommand{\Authormarginmark}[1]{}
\fi

\definecolor{forestgreen(traditional)}{rgb}{0.0, 0.27, 0.13}

\newcommand{\Rnote}{\Authornote{M}}

\ifnum\showfixme=0

\fi

\usepackage{boxedminipage}

\newcommand{\norm}[1]{\lVert#1\rVert}

\newcommand{\textparen}[1]{\text{(#1)}}

\ifx\because\undefined
\newcommand{\because}[1]{\textparen{because #1}}
\else
\renewcommand{\because}[1]{\textparen{because #1}}
\fi

\newcommand\bdot\bullet

\ifx\mathds\undefined 

\else
}
\fi

\DeclareMathOperator{\poly}{poly}

\DeclareMathOperator{\argmax}{argmax}

\DeclareMathOperator{\supp}{supp}

\newcommand{\Lovasz}{Lov\'asz\xspace}

\newcommand{\Z}{\mathbb Z}

\newcommand{\R}{\mathbb R}
\newcommand{\C}{\mathbb C}

\newcommand{\cD}{\mathcal D}

\newcommand{\cF}{\mathcal F}

\renewcommand{\leq}{\leqslant}
\renewcommand{\le}{\leqslant}
\renewcommand{\geq}{\geqslant}
\renewcommand{\ge}{\geqslant}

\ifnum\showdraftbox=1
\newcommand{\draftbox}{\begin{center}
  \fbox{%
    \begin{minipage}{2in}%
      \begin{center}%
          \Large\textsc{Working Draft}\\%
        Please do not distribute%
      \end{center}%
    \end{minipage}%
  }%
\end{center}
\vspace{0.2cm}}
\else
\newcommand{\draftbox}{}
\fi

\let\epsilon=\varepsilon

\numberwithin{equation}{section}

\newcommand\MYcurrentlabel{xxx}

\newcommand{\MYstore}[2]{%
  \global\expandafter \def \csname MYMEMORY #1 \endcsname{#2}%
}

\newcommand{\MYload}[1]{%
  \csname MYMEMORY #1 \endcsname%
}

\newcommand{\MYnewlabel}[1]{%
  \renewcommand\MYcurrentlabel{#1}%
  \MYoldlabel{#1}%
}

\newcommand{\MYdummylabel}[1]{}

\newcommand{\torestate}[1]{%
  \let\MYoldlabel\label%
  \let\label\MYnewlabel%
  #1%
  \MYstore{\MYcurrentlabel}{#1}%
  \let\label\MYoldlabel%
}

\newcommand{\restatetheorem}[1]{%
  \let\MYoldlabel\label
  \let\label\MYdummylabel
  \begin{theorem*}[Restatement of \prettyref{#1}]
    \MYload{#1}
  \end{theorem*}
  \let\label\MYoldlabel
}

\newcommand{\restatelemma}[1]{%
  \let\MYoldlabel\label
  \let\label\MYdummylabel
  \begin{lemma*}[Restatement of \prettyref{#1}]
    \MYload{#1}
  \end{lemma*}
  \let\label\MYoldlabel
}

\newcommand{\restateprop}[1]{%
  \let\MYoldlabel\label
  \let\label\MYdummylabel
  \begin{proposition*}[Restatement of \prettyref{#1}]
    \MYload{#1}
  \end{proposition*}
  \let\label\MYoldlabel
}

\newcommand{\restatefact}[1]{%
  \let\MYoldlabel\label
  \let\label\MYdummylabel
  \begin{fact*}[Restatement of \prettyref{#1}]
    \MYload{#1}
  \end{fact*}
  \let\label\MYoldlabel
}

\newcommand{\restate}[1]{%
  \let\MYoldlabel\label
  \let\label\MYdummylabel
  \MYload{#1}
  \let\label\MYoldlabel
}

\newcommand{\addreferencesection}{
  \phantomsection
  \addcontentsline{toc}{section}{References}
}

\newcommand{\e}{\epsilon}
\newcommand{\eps}{\epsilon}

\let\origparagraph\paragraph
\renewcommand{\paragraph}[1]{\origparagraph{#1.}}

\allowdisplaybreaks

\usepackage{paralist}

\usepackage{comment}

\usepackage{braket}

\usepackage{relsize}
\usepackage[font=footnotesize]{caption}
\usepackage{appendix}

\DeclareUrlCommand\email{}

\DeclareMathOperator{\zo}{\{0,1\}}

\DeclareMathOperator*{\pE}{\tilde{\mathbb E}}

\makeatletter
\def\@aparagraph[#1]#2{\paragraph[#1]{#2\@addpunct{.}}}
\def\@bparagraph#1{\paragraph*{#1\@addpunct{.}}}
\makeatother

\newcommand{\ones}{\mbox{\boldmath $1$}}
\newcommand{\zeros}{\mbox{\boldmath $0$}}
\newcommand{\xx}{\mbox{\boldmath $x$}}
\newcommand{\yy}{\mbox{\boldmath $y$}}
\renewcommand{\e}{\mbox{\boldmath $e$}}
\renewcommand{\e}{\mbox{\boldmath $e$}}

\newcommand{\tmu}{\tilde{\mu}}
\newcommand{\bS}{{\bar{S}}}
\newcommand{\NE}{\mathrm{NE}}
\renewcommand{\P}{\mathcal{P}}

\newcommand{\OV}{\mathsf{OV}}
\newcommand{\epshat}{\hat{\epsilon}}

\newcommand{\hxx}{\hat{\xx}}
\newcommand{\hyy}{\hat{\yy}}
\newcommand{\rh}{\hat{R}}
\newcommand{\ch}{\hat{C}}
\newcommand{\CF}{{\mathcal{F}}}
\newcommand{\CS}{{\mathcal{S}}}
\newcommand{\CG}{{\mathcal{G}}}
\newcommand{\ud}[1]{{{\mathcal{U}_{#1}}}}
\newcommand{\et}[1]{{e_{#1}^\top}}

\newcommand{\vv}{{\mathbf{v}}}

\newcommand{\tR}{{\tilde{R}}}
\newcommand{\txx}{{\tilde{\xx}}}
\newcommand{\tx}{{\tilde{x}}}
\newcommand{\tyy}{{\tilde{\yy}}}
\newcommand{\tC}{{\tilde{C}}}
\newcommand{\Supp}{{\mbox{\textsf{supp}}}}
\newcommand{\nice}{\textsf{SoSHard }} 
\newcommand{\good}{\textsf{EnumHard }} 
 
\newcommand{\BR}{{\mbox{\textsf{BR}}}}
\newcommand{\eBR}{{\mbox{$\eps$-\textsf{BR}}}}
\newcommand{\ehBR}{{\mbox{$\epshat$-\textsf{BR}}}}
\newcommand{\NP}{\mathrm{NP}}
\newcommand{\coNP}{\mathrm{co}\text{-}\mathrm{NP}}

\newcommand{\PPAD}{\mathrm{PPAD}}
\newcommand{\WNE}{\textsf{WNE }}

\newcommand{\alert}[1]{{\color{black} #1}}

\title{Sum-of-Squares meets Nash: Optimal Lower Bounds for Finding \emph{any} Equilibrium}

\author{Pravesh K. Kothari \thanks{Princeton University/IAS, Princeton. Supported by a Schmidt Foundation Fellowship and NSF Grant \# CCF-1412958} \and Ruta Mehta \thanks{UIUC, Urbana-Champaign.}}

\date{}

\begin{document}

\maketitle
 \draftbox
\thispagestyle{empty}


\begin{abstract}
Several works have shown unconditional hardness (via \emph{integrality gaps}) of computing equilibria using strong hierarchies of convex relaxations. Such results however only apply to the problem of computing equilibria that optimize a certain objective function and not to the (arguably more fundamental) task of finding \emph{any} equilibrium. 

We present an algorithmic model based on the sum-of-squares (SoS) hierarchy that allows escaping this inherent limitation of integrality gaps. In this model, algorithms access the input game only through a relaxed solution to the natural SoS relaxation for computing equilibria. They can then adaptively construct a list of candidate solutions and invoke a \emph{verification} oracle to check if any candidate on the list is a solution. This model captures most well-studied approximation algorithms such as those for Max-Cut, Sparsest Cut and Unique-Games.

The state-of-the-art algorithms for computing exact and approximate~\cite{LMM} equilibria in two-player, n-strategy games are captured in this model and require that at least one of i) size ($\approx$ running time) of the SoS relaxation or ii) the size of the list of candidates, be at least $2^{\Omega(n)}$ and $n^{\Omega(\log{(n)})}$ respectively. Our main result shows a lower bound that matches these upper bound up to constant factors in the exponent.

This can be interpreted as an unconditional confirmation, in our restricted algorithmic framework, of Rubinstein's recent \emph{conditional} hardness~\cite{Rub} for computing approximate equilibria.

Our proof strategy involves constructing a family of games that all share a common sum-of-squares solution but every (approximate) equilibrium of one game is far from every (approximate) equilibrium of any other game in the family. Along the way, we strengthen the unconditional lower bound against enumerative algorithms for finding approximate equilibria due to Daskalakis-Papadimitriou~\cite{DP} and the classical hardness result for finding equilibria maximizing welfare due to Gilboa-Zemel~\cite{GZ}. 

\end{abstract}

\clearpage

\ifnum\showtableofcontents=1
{
\tableofcontents
\thispagestyle{empty}
 }
\fi

\clearpage

\setcounter{page}{1}

\newcommand{\on}{\{-1,1\}}
\section{Introduction}
\label{sec:intro}

\centerline{\em How hard is it to compute equilibria in two player finite games? }

This foundational question has been a driving force of research in algorithmic game theory for almost three decades. Beginning with the formalization of the complexity class $\PPAD$, a systematic investigation~\cite{PPAD,winlose,Savani-Stengel,DP,DGP,CDT} of this question eventually led to the elegant result~\cite{DGP, CDT} establishing that computing equilibria ($\NE$) in two-player games is $\PPAD$-complete. 

Algorithmic progress has thus relied on making additional assumptions on the structure of the games or relaxing the notion of equilibria~\cite{DP,LMM,winlose,KT,Barman,TS,AGMS,BR}. Perhaps the biggest success in this direction has been the celebrated quasi-polynomial time approximate scheme of Lipton, Markakis and Mehta~\cite{LMM} for approximate $\NE$. In a major recent breakthrough, Rubinstein~\cite{Rub} showed that we cannot obtain any significant asymptotic improvement over this algorithm assuming a strong, new conjecture: solving the $\PPAD$-complete ``End-of-the-line'' problem requires exponential time.

We thus obtain a fairly complete picture of complexity of $\NE$ in general games if we believe strong enough fine-grained conjectures about the complexity of PPAD. This work is focused on the central goal of obtaining reliable, independent evidence (given that a resolution currently appears out of reach!) for such fine-grained complexity assumptions.

\paragraph{Unconditional Hardness for $\NE$?} 
A standard approach for investigating such questions attempts to build a \emph{conditional} theory of hardness based on (more believable) conjectures. This approach has the attractive feature of understanding the limitations of any algorithm for computing $\NE$. However, since it is known that $\NE$ cannot be NP-hard unless $\NP = \coNP$~\cite{MP}, such a plan encounters serious obstacles.

The other influential approach builds \emph{unconditional} hardness results for strong but restricted class of relevant algorithmic techniques. Such results have largely focused on algorithmic techniques based on systematic hierarchies of linear/semi-definite programs (LP/SDP), such as Sherali-Adams, \Lovasz-Schrjiver and Sum-of-Squares. The focus of this work is investigating the sum-of-squares SDP hierarchy for computing $\NE$.

The sum-of-squares SDP hierarchy generalizes spectral methods and linear programming and is formally the strongest known hierarchy of general, convex-programming based algorithmic schemes. It captures the state-of-the-art algorithms for many fundamental algorithmic problems \cite{MR3424199-Arora15,MR2932723-Barak11,DBLP:journals/eccc/GuruswamiS11,MR2535878-Arora09} including the celebrated QPTAS for computing approximate equilibria in game~\cite{HNW}, breaks many known hard instances of basic problems in combinatorial optimization \cite{MR2961513-Barak12,MR3202997-ODonnell12}, and has been remarkably successful in algorithm design for both worst-case and average case problems \cite{BarakKS2017,DBLP:conf/stoc/BarakKS14,DBLP:journals/corr/MaSS16,MR3388192-Barak15,DBLP:conf/colt/BarakM16,KS17a,KS17b,DBLP:conf/eurocrypt/BarakBKK18,HL17,KKM18}. 

Given the power of the technique, lower bounds against it can be credible indicators of computational hardness and have been successfully used as such especially in areas where standard conditional hardness results are infeasible. For example, some of the strongest known evidence of hardness of many fundamental \emph{average-case} problems such as Planted Clique, Refuting Random Constraint Satisfaction Problems and Maximizing Random Polynomials (or, ``Tensor Principal Component Analysis'') comes from strong lower bounds~\cite{DBLP:journals/eccc/BarakHKKMP16,MR3473335-Allen15,DBLP:conf/stoc/KothariMOW17,DBLP:journals/corr/RaghavendraRS16,DBLP:conf/colt/HopkinsSS15}  against the sum of squares algorithm.  

However, a priori, this approach also suffers from a similar foundational problem as the first one. Lower bounds for convex relaxations are usually formalized as \emph{integrality gaps}. Integrality gaps show impossibility of approximating a certain objective function on the underlying solution space with the given convex relaxation. However, $\NE$ is a search problem without any pre-defined objective function. Restricting to the natural decision version trivializes the problem as every two player finite game has an equilibrium. Thus, integrality gaps do not provide means of establishing hardness for the problem.

Prior works~\cite{CS,Hazan-Krauthgamer,BKW,Deligkas-Fearnley-Savani} dealt with this issue by forcing some natural objective function on the space of equilibria and establishing the hardness of finding an equilibrium that optimizes the chosen objective. For appropriate choices of such objectives, finding (or even approximating) the best equilibria is known to be NP-hard~\cite{GZ,CS}. However, such a result doesn't allow distinguishing whether the hardness is of finding \emph{any} equilibrium or one that optimizes the chosen objective function. 

The main contribution of this work is a new approach to circumvent such issues and establish unconditional hardness for finding \emph{any} equilibrium using the sum-of-squares method by relying on \emph{rounding gaps} instead of \emph{integrality gaps}.

\paragraph{Hardness via Rounding Gaps} The conceptual idea underlying our framework is quite simple. A convex relaxation, such as the SoS SDP, returns a certain relaxed solution for the problem of our interest. A relaxed solution will not generally be an actual solution so we thus must use an additional, second step, usually known as \emph{rounding}, that transforms a relaxed solution into an actual equilibrium. Given that any algorithm for $\NE$ based on such a convex relaxation must go through the additional rounding step, we will  show hardness results for both the steps above combined.

There's an important issue, however. At the outset, the rounding algorithm might just ignore the convex relaxation's solution and just find an equilibrium for the underlying game from scratch. Without settling the complexity of PPAD, thus, we cannot hope to prove a lower bound against such ``rounding'' algorithms. Thus, we must restrict our rounding to procedures that really ``use'' the solution to the convex relaxation and do not ``cheat'' by ignoring it.

\paragraph{Oblivious Rounding and Verification Oracles} 
An \emph{oblivious} rounding algorithm takes input the relaxed solution (generated by the convex relaxation) and outputs a true solution with the crucial restriction that the output be a function of \emph{only} the relaxed solution. In particular for us, an oblivious rounding algorithm for computing $\NE$ accesses the underlying payoff matrices of the input game only indirectly via the solution to the SoS SDP.

This class of roundings was first formalized in a work of Feige, Feldman and Talgam-Cohen~\cite{DBLP:conf/approx/FeigeFT16} in the context of certain mechanism design problems (see a longer discussion in Section \ref{sec:related-work}). While such rounding algorithms might appear restrictive at the outset, a simple inspection reveals that, in hindsight, some of the most powerful SDP rounding algorithms in the literature are in fact oblivious! This includes, for instance, the famous algorithms for Max-Cut~\cite{DBLP:conf/stoc/GoemansW01}, Sparsest Cut~\cite{MR2535878-Arora09} and unique games/small-set-expansion~\cite{MR2932723-Barak11}. Finally, a recent work of Harrow-Natarajan and Wu~\cite{HNW} shows that the guarantees of the Lipton-Markakis-Mehta~\cite{LMM} algorithm for computing approximate equilibria can be matched by an \emph{oblivious} rounding based algorithm. 

While oblivious rounding captures many famous rounding algorithms, there are a few notable exceptions that do not fit into this framework. This includes, for e.g., the rounding method for arbitrary Constraint Satisfaction Problems~\cite{MR2648437-Raghavendra09} and the recent works on polynomial optimization and its variants~\cite{DBLP:conf/stoc/BarakKS14}. Such works construct a list of candidate solutions based on the solution to the relaxation and show that one of the candidates must in fact be a true solution. Indeed, such an idea in particular captures the spectacularly successful ``enumerative techniques'' in computing equilibria for various classes of games: cleverly construct a small search space that contains a solution and then brute-force search over it to find the solution, e.g.,~\cite{LMM,KT,DP,MR3210829-Alon13}. Perhaps the most famous example of the usage of this idea is in fact the algorithm in~\cite{LMM} discussed above.

We thus allow our model to (possibly adaptively) generate a list of candidates and check if one of them is a solution. This gives us the final class of roundings we work with in this paper: \emph{Oblivious Roundings with Verification Oracle} ($\OV$). We measure the cost of the rounding in two parameters: the running time of the relaxation and the number of candidate solutions generated. Notably, we do not restrict the running time of the algorithm that constructs such a list from the solution to the relaxation - in that sense, our model is information theoretic. We also allow certain strengthening of the verification oracle but defer the details to the next section. 

The resulting model is a strict strengthening of the \emph{oblivious algorithms} framework used by Daskalakis-Papadimitriou~\cite{DP} to show optimality of the QPTAS for approximate equilibria (longer discussion in Section \ref{sec:related-work}). Our lower bounds thus immediately strengthen theirs.

Finally, we must note that our algorithmic model is strictly weaker than integrality gaps for situations where they make sense. This is because by construction, they can never capture the class of all rounding algorithms for a given convex relaxation. More practically,  there are important rounding schemes that are not captured in our framework - for e.g., iterative rounding techniques such as those used for the facility location problem~\cite{DBLP:journals/iandc/Li13}. At present, we do not know a clean extension of our framework that can capture such rounding methods.

\paragraph{Known Algorithms} It is not hard to show that $2^{O(n)}$ queries to the verification oracle (and ``0'' convex relaxation cost) can find exact equilibria. The QPTAS of~\cite{LMM} implies an algorithm with $n^{O(\log{(n)})}$ queries to find constant approximate equilibria (and ``0'' convex relaxation cost). Finally, Harrow-Natarajan-Wu~\cite{HNW} show that the same result can be obtained via an algorithm with $n^{O(\log{(n)})}$-convex relaxation cost and ``0'' queries to the verification oracle. 

We are now ready to summarize our main hardness results for computing equilibria in this framework.

\paragraph{Summary of Results} 
Our main results establish that the known algorithms for computing equilibria, both exact and approximate, are optimal, up to polynomial factors in our model. Specifically, we show that for two-player games with $n$ strategies and all payoffs in $[-1,1]$, 
\begin{enumerate}
\item for any $\epsilon \leq \Theta(1/n^4)$, there's no algorithm that uses a $2^{o(n)}$-time SoS relaxation to construct a $2^{o(n)}$-size list of candidates to compute an $\epsilon$-approximate $\NE$. 
\item for some constant $\epsilon > 0$, there's no algorithm that uses a $n^{o(\log{(n)})}$-time SoS relaxation to construct a list of size $n^{o(\log{(n)})}$ to compute an $\epsilon$-approximate equilibrium. 
\end{enumerate}
Our results can be seen as unconditional confirmation, in our restricted algorithmic framework, of the recent result of Rubinstein~\cite{Rub}. 

Since ours is a stronger model of computation, our lower-bounds strictly improve upon those of Daskalakis-Papadimitriou~\cite{DP}. Along the way, we also strengthen the lower bound in their more restricted model (see Section~\ref{sec:our-results}). We also obtain a ``gapped'' (approximation hardness) version improving on the classical result of Gilboa-Zemel~\cite{GZ} (see Section~\ref{sec:our-results}).

We give a brief outline of what follows next before the technical sections. In Section~\ref{sec:setup}, we present a formal description of our algorithmic model and state precise versions of our results. We also discuss some of the improvements to previous work that we obtain along the way. In Section~\ref{sec:techoverview}, we give an overview of the main technical ideas required in our proofs. Finally, in Section \ref{sec:related-work}, we briefly discuss important related work. 

\section{Algorithmic Model and Statements of Results} 
\label{sec:setup} 
We first set up the standard terminology for talking about games and equilibria. 

A game $G$ between Alice and Bob with $n$ strategies for both is described by two payoff matrices $R$ and $C$ in $[-1,1]^{n\times n}$,\footnote{Adding dummy strategies to make the number of strategies equal, and normalizing payoffs to lie in $[-1,1]$ is without loss of generality.} 
where Alice plays rows and Bob plays columns. A mixed-strategy (now on called strategy) for either of the players is an element of $\Delta_n = \{x \in [0,1]^n \mid \sum_{i = 1}^n x_i = 1\}$ and encodes a probability distribution on the space of $n$ pure strategies $\{1,\dots,n\}$ denoted by $[n]$ now on.

We write $\e_i  \in \Delta_n$, for the $i^{th}$ pure strategy. When Alice plays $\xx\in \Delta_n$ and Bob $\yy \in \Delta_n$, the expected \emph{payoffs} for the two players are $\xx^{\top} R \yy$ and $\xx^{\top} C \yy$, respectively. 

A Nash Equilibrium ($\NE$) is a strategy pair $(\xx,\yy)$ such that no player gains by deviating unilaterally. That is,
\begin{align}
\xx^{\top} R \yy \geq \e_i^{\top} R \yy  \text{ for all } i \in [n]; \ \ \ \xx^{\top} C \yy \geq \xx^{\top} C \e_j  \text{ for all } j \in [n]. \label{eq:NEchar}
\end{align} 
At an $\epsilon$-approximate NE ($\eps$-$\NE$) no player gains by more than $\eps$ by deviating unilaterally: 
\begin{align}
\xx^{\top} R \yy \geq \e_i^{\top} R \yy -\epsilon  \text{ for all } i \in [n]; \ \ \ \   
\xx^{\top} C \yy \geq \xx^{\top} C \e_j -\epsilon \text{ for all } j \in [n]. \label{eq:approx-NEchar}
\end{align}

\paragraph{Sum-of-Squares Method and \emph{Pseudo-equilibria}}
The sum-of-squares method is a sequence of increasingly tight SDP relaxations, indexed by an integer parameter $d$, for the problem of finding a solution to a system of polynomial inequalities in real valued variables. We provide a brief overview of this method specialized to computing (exact/approximate) equilibria here and point the reader to the lecture notes \cite{sos-notes} for further details on method and its applications.

Central to the sum of squares method is the notion of a pseudo-distribution that generalizes probability distributions. 

\begin{definition}[Pseudo-distribution]\label{def:pe}
A degree $d$ pseudo-distribution is a finitely supported signed measure $\tilde{\mu}$ on $\R^n$ such that the associated linear functional (pseudo-expectation) $\pE$ that maps any function $f:\R^n \rightarrow \R$ to $\pE[f] = \sum_{x: \tilde{\mu}(x) \neq 0} \tilde{\mu}(x) f(x)$ satisfies the following properties:
\begin{enumerate}
\item \textbf{Normalization: } $\pE[1] = 1$ or equivalently, $\sum_{x: \tilde{\mu}(x) \neq 0} \tilde{\mu}(x) = 1,$ and
\item \textbf{Positivity: } $\pE[q^2] \geq 0$ for every degree $\leq d/2$ polynomial $q$ on $\R^n.$
\end{enumerate}
A pseudo-distribution $\tilde{\mu}$ is said to satisfy a polynomial inequality constraint $q \geq 0$ if for every polynomial $p$, $\pE[p^2 q] \geq 0$ whenever $\deg(p^2q) \leq d.$
\end{definition}

It's not hard to show that any degree $\infty$ pseudo-distribution is a probability distribution on $\R^n.$

Since $\pE_{\tmu}$ is a linear operator, it is completely specified by its action on any basis of $\P_d.$ The monomial $x^S = \Pi_{i}x_i^{S_i}$ for any $S \in \Z^n$, $\sum_{i = 1}^n S_i \leq d$ (the total degree) form a basis for $\P_d.$ Let $(1,x) \in \R^{n+1}$ denote the vector with first coordinate $1$ and last $n$ coordinates matching $x \in \R^n.$ Then, the order-$d$ tensor $\pE_{\tmu}[(1,x)^{\otimes d}]$ has entries equal to pseudo-moments of $\tmu$ and thus completely describes the pseudo-expectation $\pE_{\tmu}.$

When $\tmu$ is a probability measure on $\R^n$, the associated pseudo-expectation is an actual expectation operator. Key to the utility of the notion of pseudo-distributions is the following classical fact. 

\begin{fact}[Lasserre, Shor, Parillo, Nesterov \cite{DBLP:journals/siamjo/Lasserre01,MR1778235-Nesterov00,parrilo2000structured,MR1034072-Shor89}] \label{fact:sos-algorithm}
Let $g_1, g_2, \ldots, g_r, h_1, h_2, \ldots, h_q$ be polynomials of degree at most $d$. Then, the convex set 
\[
\C_d = \{ \pE_{\tmu}[(1,x)^{\otimes d}] \mid \tmu \text{: deg $d$, satisfies } g_1 = 0, g_2 = 0, \ldots, g_r = 0, h_1 \geq 0, h_2 \geq 0, \ldots, h_s \geq 0\}
\] has an $n^{O(d)}$-time weak separation oracle in the sense of \cite{MR625550-Grotschel81}. 
\end{fact}

\begin{definition}[Sum-of-Squares Algorithm]
The degree $d$ sum-of-squares algorithm takes input a system of polynomial constraints $g_1 = 0, g_2 = 0, \ldots, g_r = 0, h_1 \geq 0, h_2 \geq 0, \ldots, h_s \geq 0$ each of degree at most $d$ and either outputs ``infeasible'' or returns a degree $d$ pseudo-expectation $\pE_{\tmu}$ satisfying Definition \ref{def:pe}. Fact \ref{fact:sos-algorithm} implies that the degree $d$ SoS algorithm runs in time $n^{O(d)}.$
\end{definition}

Pseudo-distributions naturally suggest a relaxation of the notion of equilibrium, which we call, \emph{pseudo-equilibrium}. The advantage is that unlike equilibria, pseudo-equilibria of degree $d$ are \emph{efficiently} computable using Fact \ref{fact:sos-algorithm}. It might be helpful to note that degree $2$ pseudo-equilibria (without the semi-definite constraints) correspond to the well-studied notion of \emph{correlated} equilibria.  

\begin{definition}[Degree $d$ Pseudo-equilibrium]
\label{def:pe}
Given a two player game $(R,C)$, a degree $d$ pseudo-equilibrium for $(R,C)$ is a degree $d$ pseudo-distribution on strategy profiles $(\xx,\yy)$ satisfying the quadratic polynomial inequality constraints in \eqref{eq:NEchar}. A degree $d$, $\epsilon$-approximate pseudo-equilibrium is a degree $d$ pseudo-distribution over strategy profiles $(\xx,\yy)$ satisfying the quadratic polynomial constraints in \eqref{eq:approx-NEchar}.
\end{definition}

\subsection{Rounding Framework for Equilibria} 
In this section, we present our algorithmic framework based on restricted roundings of sum-of-squares relaxations.
While the framework naturally generalizes to any polynomial feasibility problems, we will focus only on (exact) $\NE$ here.

Let $\Psi(G)$ be the system of polynomial inequalities in $\xx,\yy \in \R^n$ parameterized by a game $G$ represented by $R,C \in \R^{n\times n}$ described in \eqref{eq:NEchar}. 

A rounding algorithm takes a solution to the SoS relaxation for $\Psi(G)$ and outputs a candidate equilibrium strategy profile $(\xx^{*}, \yy^{*})$. Before presenting our restricted rounding framework, we formally define what we mean by a rounding algorithm.

\begin{definition}[Rounding Algorithm]
A degree $d$ rounding algorithm for $\Psi(G)$ takes input a game $G$ and a degree $d$ pseudo-expectation satisfying the NE constraints of $G$ and outputs a solution to $\Psi(G)$.
\end{definition}
\begin{remark}
As discussed, this definition captures all algorithms as one could simply ignore the degree $d$ pseudo-expectation. Notice that we place no restriction on the running time of the rounding algorithm at all. 
\end{remark}

Next, we present a meaningful restriction of arbitrary roundings that forces the algorithm to ``use'' the solution to the relaxation. We do this by simply disallowing the access to the game $G$ itself. 

\begin{definition}[Oblivious Rounding]
\label{def:alg-or}
A degree $d$ oblivious rounding algorithm is a rounding algorithm that does not get access to the instance (i.e., game) $G$ itself. The cost of an oblivious rounding algorithm is the degree of the SoS solution it accesses.

\end{definition}
\begin{remark}
Two comments are in order:
\begin{enumerate}
\item Observe that we do not restrict the running time of an oblivious algorithm. In that respect, this model and the resulting lower bounds are information theoretic.
\item This definition was first proposed in~\cite{DBLP:conf/approx/FeigeFT16} for optimization problems (ours is a ``feasibility'' problem as it has no objective) who gave a characterization of problems where oblivious rounding algorithms achieved the underlying integrality gap. 
\end{enumerate}

\end{remark}
While this may appear like a restricted framework, it nonetheless captures several famous algorithms obtained via convex relaxation + rounding paradigm such as those for Max-Cut, Sparsest-Cut, Unique Games etc.

To allow our rounding algorithm more power and capture a longer list of existing algorithms, we allow the oblivious rounding algorithm generate a list of candidate solutions instead of a single one and then check if one of them is indeed a solution using access to a \emph{verification oracle}. This allows us to $(i)$ capture most known rounding algorithms mentioned above with the exception of techniques such as iterative rounding $(ii)$ provide our model the ability to simulate ``enumeration'' over a restricted search space - a widely used technique in computing equilibria for various classes of games.

To formalize this model, we will allow the algorithm to access the game through a \emph{verification oracle} in addition to a degree $d$ pseudo-equilibrium.

\begin{definition}[Verification Oracle]
\label{or1}
A verification oracle for an instance (i.e., game) $G$ takes input a candidate solution (i.e., a strategy profile $(\xx,\yy)$) and correctly outputs ``accept'' if the candidate is a true solution (i.e, an equilibrium) and ``reject'' otherwise.
\end{definition}

We can now augment the oblivious rounding algorithms with a verification oracle access to the underlying instance (game) $G$.
\begin{definition}[Oblivious Rounding with Verification Oracle $\OV$]
\label{def:alg-ov}
A degree $d$, $q$-query \emph{oblivious rounding} algorithm \emph{with verification oracle} is a degree $d$ oblivious rounding algorithm that, in addition, accessed a verification oracle for the underlying instance $G$ at most $q$ times.  
\end{definition}
\begin{remark}
Three comments are in order, again. 
\begin{enumerate}
\item As in the case of oblivious rounding algorithms, note that we do not restrict/measure the running time of the algorithm.
\item Observe that the model allows ``adaptivity'' - the candidate solutions can be generated after looking at the replies of the verification oracle for previous candidates. 
\item This model strictly strengthens the ``Oblivious algorithms" model studied in~\cite{DP}. Oblivious algorithms, in the sense of ~\cite{DP} are simply degree $0$ $\OV$ algorithms. That is, $\OV$ algorithms that get access only to the verification oracle and neither the game nor the pseudo-equilibrium. 
\end{enumerate}
\end{remark}

Finally, for the special case of equilibria, we can allow an even stronger access to the underlying game through a more powerful oracle. Observe that if we fix one of $\xx$ or $\yy$ in \eqref{eq:NEchar} and \eqref{eq:approx-NEchar}, then the resulting system of constraints is linear in the other variables. Thus, one can check feasibility of this resulting system in polynomial time. 

Motivated by this, we can define an oracle that takes input a ``partial strategy'' - either one of $\xx$ or $\yy$ - and check if there's an equilibrium that agrees with the given partial strategy. This definition, unlike all others, is not meaningful for arbitrary polynomial systems $\Psi(G)$.

\begin{definition}[Partial Verification Oracle]
\label{or2}
A partial verification oracle for a game $G$ takes input one of $\xx,\yy$ and outputs ``accept'' if there's an equilibrium for $G$ that agrees with the input and ``reject'' otherwise.
\end{definition}
\begin{remark}
It is important to observe that a partial verification oracle is a \emph{stronger} access to the underlying game - namely, it allows checking if there's an equilibrium matching only one of the player's strategies. Indeed, as we will note in the discussion of our results, the lower bound of Daskalakis-Papadimitriou does \emph{not} hold for algorithms that have this stronger access to game via a Partial verification oracle. 

As a by product of our proof, we obtain a strengthening of their result that holds against oblivious algorithms that get access to the stronger partial verification oracle. 
\end{remark}

\subsection{Our Results} \label{sec:our-results}

Our main results are optimal lower bounds for $\OV$ roundings of Sum-of-Squares for computing exact and approximate equilibria. Throughout this section, we refer to a two-player game with $n$ strategies each and all payoffs bounded in $[-1,1]$ by just a ``game''.

Our result for computing exact equilibria is actually robust and allows a certain small but non-trivial relaxation of the notion of equilibria. More precisely, we show that:

\begin{theorem}[Exponential Lower Bound for $\NE$] \label{thm:main-exact}
Suppose there is degree $d$, $q$ query $\OV$ rounding for finding $\epsilon = \Theta(1/n^4)$-approximate equilibria in games. Then, either $d = \Omega(n)$ or $q = 2^{\Omega(n)}.$ Further, the same results holds for $\OV$ algorithms with the stronger partial verification oracle. 
\end{theorem}

This follows from a construction of a family of games with properties captured in the following theorem:
\begin{theorem} [Exponentially large hard family]
For every $n$ large enough, there's a family of $\Gamma = 2^{\Theta(n)}$ games $\{G_i=(A_i,B_i)\}_{i = 1}^{\Gamma}$ with \alert{$n$} pure strategies for both players and all payoffs in $[-1,1]$ such that:

\noindent$(1.)$ \textbf{Completeness:} There exists a degree $\Theta(n)$, shared pseudo-equilibrium for every $G_i$ simultaneously. 

\noindent$(2.)$ \textbf{Soundness:} For any $i\neq j$ if $(\xx,\yy)$ and $(\xx',\yy')$ are \alert{$\Omega(1/n^4)$-$\NE$} of games $G_i$ and $G_j$ respectively then $\xx\neq \xx'$ and $\yy \neq \yy'$. 
\label{athm:exp-lower-bound-technical-intro}
\end{theorem}

To see why this implies Theorem~\ref{thm:main-exact}, choose an input game $G$ uniformly at random from the family described by the Theorem~\ref{athm:exp-lower-bound-technical-intro}. Then, any (potentially randomized) algorithm that succeeds with $2/3$ probabilty in computing an $\epsilon$-$\NE$ in $G$ requires either $d = \Omega(n)$ or $q = 2^{\Omega(n)}$.
This is because, using property (2) above, such an algorithm must uniquely and correctly identify the input game with probability at least $2/3$. The SoS relaxation can just output the shared pseudo-equilibrium giving zero information about the input game. Further, each ``reject'' answer from the (partial) verification oracle rules out exactly one possibility. Thus, the $\OV$ algorithm must use $2^{\Omega(n)}$ queries.

In establishing this result, we prove a ``gapped'' version of the classical hardness reduction of Gilboa-Zemel from $k$-clique to finding \emph{well-supported} equilibria with payoff at leaset $\delta>0$. While in a Nash equilibrium, each player's strategy must put a non-zero probability only on the pure strategies with maximum payoffs, in an $\epsilon$-well-supported $\NE$, the support of any player's strategy can contain pure strategies that have an $\epsilon$ additively smaller payoff compared to the maximum (see Section~\ref{sec:techoverview} for a formal definition). Well-supported equilibria are a weaker solution concept than equilibria and thus, this is technically stronger than one that shows a gapped hardness for finding $\NE$. 

\begin{theorem}
There's a polynomial time reduction that takes input a graph $H$ on $n$ vertices and parameters $k\le n, \epsilon = \Omega(1/n^2)$ and produces a two-player game $G$ with $O(n)$ strategies for each player with all payoffs in $[-1,1]$ such that (i) if $H$ has a $k$-clique there exists a $\NE$ for $G$ with payoff $\delta>0$, and $(ii)$ if $H$ has no $k$-Clique then all $\epsilon$-well-supported $\NE$ of $G$ have payoffs at most $(\delta - \eps)$. 

\end{theorem}

Our second main result is a quasi-polynomial lower bound on $\OV$ algorithms for finding constant-approximate equilibria.
\begin{theorem} \label{thm:main-approx}
Suppose there exists a degree $d$, $q$-query algorithm for computing $\epsilon$-approximate equilibria in games for some small enough $\epsilon = \Theta(1)$. Then, either $d = \Omega(\log{(n)})$ or $q = n^{\Omega(\log{(n)}}.$ Further, the same result continues to hold for $\OV$ algorithms that use the stronger partial verification oracle.
\end{theorem}

As above, we obtain this theorem by giving a construction of a hard family of games:
\begin{theorem} [Quasi-polynomial Hardness for $O(1)$-$\NE$]
For every $n$ large enough, there's a family of $\Gamma = n^{\Omega(\log{(n)})}$ two-player games $\{G_i=(R_i,C_i)\}_{i = 1}^{\Gamma}$ with $n$ strategies and all payoffs in $[-1,1]$ s.t.:

\noindent$(1.)$ \textbf{Completeness:} All $G_i$'s share a common degree $\Theta(\log{(n)})$ pseudo-equilibrium.

\noindent $(2.)$ \textbf{Soundness:} There exists an $\epsilon=O(1)$, such that for any $i,j$ and any pair of $\epsilon$-$\NE$ in $G_i$ and $G_j$, say $(\xx,\yy)$ and $(\xx',\yy')$ respectively, $\xx \neq \xx'$ and $\yy\neq \yy'$. 
\label{athm:approx-nash-hardness-technical}
\end{theorem}

Along the way to this result, we extend the result of Daskalakis and Papadimitriou on lower bounds for their model of oblivious algorithms. Their construction, in the language of this paper, shows lower bounds on the number of calls to the Verification Oracle~\ref{or1} described before ~\cite{DP}. We extend it to lower bound against the stronger Oracle~\ref{or2}. 
\begin{corollary}
For a given game $(R,C)$ with $n$ strategies, if an algorithm queries Oracle \ref{or2} $q$ times and outputs an $O(1)$-$\NE$ of $(R,C)$, then $q=n^{\Omega(\log{(n)})}$. 
\end{corollary}

\begin{remark} [Well-Supported $\NE$ vs $\NE$]
Three points in order. 
\begin{enumerate}

\item The notion of approximate well-supported $\NE$ has been of independent interest and studied as such (see for e.g.~\cite{CDT,DP}). Since every $\eps$-\WNE of a game is also its $\eps$-$\NE$, all our lower bounds hold for finding \WNE as well. 
\item Our proofs, in fact, yield a stronger hardness result for finding $\epsilon$-\WNE: we show that finding  $\epsilon=O(1/n^2)$-\WNE requires an exponential degree or queries to the (partial) verification oracle.  
\item One can strengthen our results to an oracle that gets another strong access to the underlying game: given the support of both player's strategies, accept if there's an equilibrium with those supports. Our lower bounds work against such verification oracles too. 
\end{enumerate}
\end{remark}

\subsection{Related Work} \label{sec:related-work}
We discuss three works that are especially relevant to this work. 

The sum-of-squares SDP hierarchy was recently studied in the context of computing equilibria by Harrow-Natarajan-Wu~\cite{HNW}. In the language of the present paper, they showed that $(i)$ there's a quasi-polynomial degree oblivious rounding  (that does not need access to the verification oracle) that matches the guarantees of the algorithm in~\cite{LMM} and $(ii)$ there's a cost\footnote{While the proof in \url{https://dspace.mit.edu/openaccess-disseminate/1721.1/109803} shows a result for a somewhat artificial objective function (see Theorem 3.11 and the proof that follows for the description of this objective), in private communication with the authors, we learned that their techniques extend to taking the objective function as social welfare with some more work.} measure on equilibria such that the degree $\Omega(\log{(n)})$ sum-of-squares algorithm has a constant-factor integrality gap in approximating it over the space of all constant approximate equilibria. While both the works show a hardness for the sum-of-squares algorithm for computing equilibria, in contrast to their work that focuses on showing an integrality gap for a certain objective measure over equilibria, we focus on the task of computing \emph{any} equilibrium. 

The notion of oblivious roundings was first studied in a work due to Feige, Feldman and Talgam-Cohen~\cite{DBLP:conf/approx/FeigeFT16}. Their work gave a neat characterization of the problems where oblivious roundings can achieve the integrality gap (for optimization, instead of feasibility problems that we study in this work) of the underlying relaxation. Their main application was for the maximum welfare problem where they showed that if the valuation functions of the agents are submodular, then, oblivious rounding achieves the integrality gap of the natural configuration LP relaxation for the problem, while, for the case of gross substitute valuations, oblivious roundings cannot achieve the integrality gap of the configuration LP for the problem. 

The work of Daskalakis-Papadimitriou~\cite{DP} studied algorithms that \emph{obliviously} (= degree 0 $\OV$ algorithms using the standard verification oracle~\ref{or1}) search the smallest possible space of strategy profiles to find an approximate equilibrium. Their main motivation was showing that among all oblivious algorithms, the Lipton-Markakis-Mehta algorithm is almost optimal for computing approximate equilibria in games. Their lower bounds, however do not apply to algorithms that use the stronger partial verification oracle~\ref{or2}. 

While this can appear somewhat technical, this difference comes from a fairly intuitive difference between our two constructions. The construction of~\cite{DP} gives a collection of $\approx n^{\log{(n)}}$ games such that every approximate equilibrium of any single game is different from every approximate equilibrium of any other. However, each game admits an approximate equilibria that agree on the strategy of one of the players. While this is enough to fool the standard verification oracle (=oblivious algorithms), it does not suffice to fool the stronger partial verification oracle.

\section*{Acknowledgement}
We thank the anonymous reviewers for pointing out typos and missed references and giving useful suggestions for better presentation. We thank an anonymous reviewer for pointing out that the algorithm in Harrow, Natarajan and Wu~\cite{HNW} is an ``oblivious rounding'' algorithm in the sense of this paper. We also thank Anand Natarajan for answering a question about~\cite{HNW}. P.K. thanks Avi Wigderson for helpful discussions, suggestions and comments during various stages of this work. 

\section{Technical Overview} \label{sec:techoverview}

Our main results establish lower bounds that match the existing upper bounds for $\OV$ algorithms with the stronge partial verification oracle for computing $\Theta(1/n^4)$ (robust, exact) and $\Theta(1)$-approximate Nash equilibrium computation.

Both our constructions will follow via three modular steps.  
\begin{enumerate}
\item Construct a family of games that are hard for degree $0$ $\OV$ algorithms - i.e., algorithms that rely only on the (partial) verification oracle. This will strengthen the result in~\cite{DP} for constant-approximate equilibria and provide a new construction for the case of (exact) $1/poly(n)$-approximate equilibria.
\item Construct a game such that approximating the equilibrium that maximizes the social welfare (sum of the payoffs) of the two players (social welfare) is hard for degree $d$ SoS algorithm. It is important to stress this point: even though our final aim is to construct a lower bound for the feasibility problem, the SoS hardness is for the \emph{optimization} version, namely, for maximizing the social welfare. As discussed in the related work, the proof in~\cite{HNW} does not give such a result at the moment. However, there are hardness results in the optimization setting (we rely on the recent work in~\cite{Deligkas-Fearnley-Savani}) based on standard conjectures such as $P \neq NP$ or Exponential Time Hypothesis in this setting which we leverage in our construction. 

\item A black-box result that allows ``stitching'' together the two kinds of hard games into a single family of games that are hard for $\OV$ algorithms. The technical heart of the proof here is establishing that the process of ``stitching'' together the two kinds of hard games does not introduce new spurious (approximate) equilibria that destroy the first or the second property above.  
\end{enumerate}

This overview is organized as follows. First, we describe the ideas in our ``stitching'' theorem as it applies to both our lower-bound results. This part holds in a fairly general setting. Second, we describe the ideas in the construction of the two kinds of games above to complete the proofs of our two lower bound results.

\subsection{Stitching Hard Games Together}\label{sec:abs}

Before proceeding, we need to define the notion of well-supported equilibria: a well-studied relaxation of the notion of equilibria~\cite{CDT,DP} that we briefly discussed in Section \ref{sec:our-results}.

For integer $n$ and any $T\subseteq[n]$, let $\ud{n,T}$ be probability vector of a distribution on $[n]$ that is uniform over $T$.

\begin{definition}[Well-Supported Equilibria]
\label{def:wne}
A strategy profile $(\xx,\yy)$ of game with payoff matrices $(R,C)$ is called an $\eps$-approximate well-supported $\NE$ (in short, $\eps$-\WNE) iff every pure strategy in the support of $\xx$ has a payoff within $\epsilon$ of the maximum given the column player plays $\yy$ and vice-versa. That is,
\begin{equation}\label{eq:wne}
\mbox{for every } i \in [n],\ \xx_i > 0 \Rightarrow e_i^{\top} R \yy\ge \max_{k\in[n]} e_k^{\top} R \yy -\epsilon  \ \ \ \mbox{ and } \ \ \ \yy_i>0 \Rightarrow \xx^{\top} C e_i \ge \max_{k\in[n]} \xx^{\top} C e_k-\epsilon.
\end{equation}
\end{definition}
\begin{remark}
Observe that any $\NE$ is an $\epsilon$-\WNE. And every $\epsilon$-\WNE is an $\epsilon$-approximate $\NE$. However, an $\epsilon$-approximate $\NE$, in general, is not an $\epsilon$-\WNE.
\end{remark}

Next, we formalize the two kinds of hard games that we discussed in the previous subsection.

Our requirement from the $\nice$ game is straightforward. We want a game where the social welfare of a pseudo-equilibrium is higher than the social welfare of every $\epsilon$-\WNE by at least $2\epsilon$.
\begin{definition}[degree $t$, $(\eps,\delta)$-\nice game]
For some $\epsilon < \delta$, a game $(R,C)$ is said to be degree $t$, $(\epsilon,\delta)$-\nice game if there's a degree $t$ pseudo-equilibrium with payoffs for each player exceeding $\delta>0$ while the social welfare (sum of the payoffs of the players) in every $\eps$-$\WNE$ is at most $2(\delta-\eps)$.  
\end{definition}

Our requirement from the hard family of games for oblivious algorithms, i.e., enumerative techniques, denoted by $\good$ games, is the following. We want a family of games with $K$ strategies each, that is indexed by a low-intersection family of subsets of $[K]$ such that all $\epsilon$-\WNE of a game indexed by $S$ are close to the uniform distribution on the strategies in $S \subseteq [K]$.

\begin{definition}[ $(\eps,\tau)$-\good game]
For some $\beta > 0$, $\cF$ be a family of subsets of $[K]$ so that for any pair $S,T \in \cF$, $|S \Delta T| \geq \beta \tau \epsilon K$, or, equivalently, $\|\ud{K,S} - \ud{K,T}\|_1 \geq \beta \tau \epsilon$. A family of games $(R_S,C_S)$ with $K$ strategies each, indexed by subsets $S \in \cF$ is said to be $(\epsilon,\tau)$-\good if any $\epsilon$-\WNE $(\xx,\yy)$ of $(R_S,C_S)$ satisfies: $\|\xx - \ud{K,S}\|_1, \|\yy-\ud{K,S}\|_1 \leq \tau$.
\end{definition}

The \nice game above gives degree $t$ SoS hardness of obtaining high payoff $\eps$-\WNE (and thus, also for finding $\NE$). Similarly, for $\beta>2$, the family of \good games constructed for each $S\in \CF$ would give lower bound of $q(\beta)$-queries to the verification oracle to find $\eps$-\WNE. Observe, however, that neither hardness hold for the notion of eqiulibria we care about: $\eps$-approximate $\NE$. 

Our goal now is to combine the above to constructions to obtain a family of games for computing $\eps$-approximate $\NE$ via $\OV$ algorithms with appropriate degree and number of queries. The major issue in combining two games is to handle equilibria that are supported on startegies of both the games. The goal of our construction is to avoid these. 

\paragraph{Construction of the Hard Family} Let $(R,C)$ be $N\times N$ game and $(R_S,C_S)$ be $K\times K$ game with all payoffs in $[-1,1]$. We will use these two constructions to obtain a family of games $(R'_S,C'_S)$ with $(N+K)$ pure strategies for each player and show 
:
\begin{enumerate}
 \item \textbf{ SoS Completeness: } any pseudo-equilibrium for the \nice game with payoff at least $\delta$ for each player is a pseudo-equilibrium for every $(R_S', C_S')$, and 
 \item \textbf{ Soundness: } every $\eps$-$\NE$ of $(R'_S, C'_S)$ is close, in L1 distance, to $\ud{S}$.
 \end{enumerate} 

We first describe the construction. Throughout this overview, it will be easier for the sake of exposition to allow payoffs in our games to exceed $1$ in magnitude - we finally just have to normalize the game to bring all payoffs in $[-1,1]$.

For every $S \in \cF$, the family of sets $\cF$ indexing the \good family of games, we construct a new game $(R'_S, C'_S)$ as follows (where $\delta$ and $-2$ denote blocks of appropriate sizes with all entries equal to $\delta$ and $-2$ respectively).
\begin{equation}\label{aeq:abs}
R'_S = \left[\begin{array}{cc} R \ & \ -2 \\ \delta \ &\ R_S \end{array}\right]\ \ \ \mbox{ and }\ \ \  C'_S = \left[ \begin{array}{cc} C \ & \ \delta \\ -2 \ & \ C_S \end{array}\right]
\end{equation}

\paragraph{Completeness} The proof of \emph{completeness} (Lemma \ref{lem:abstract-completeness}) is intuitively based on the following elementary observation. Suppose $(R,S)$ had an equilibrium $(\xx,\yy)$ with payoff at least $\delta$ for each player. Then, we note that $(\xx,\yy)$ is also an equilibrium in the new game $(R'_S,C'_S)$. This is simply because the row (respectively, column) player has no incentive to use strategies in the lower (respectively, right) block. Of course, $(R,S)$ does not have such an equilibrium, only a pseudo-equilibrium, but it turns out that the same argument as above extends to pseudo-equilibria. 

\paragraph{Soundness} The proof of \emph{soundness} requires a lengthier argument. The proof is based on two intermediate claims: $(i)$ every $\eps$-\WNE of $(R'_S, C'_S)$ is supported only on last $K$ strategies for both the players (Lemma \ref{lem:abstract}) and $(ii)$ every $\eps$-$\NE$ of $(R'_S,C'_S)$ is close (in L1 distance) to an $\eps$-\WNE. It is not hard to show soundness using these two claims, that is, that each $\eps$-$\NE$ of the game is close to (in L1 distance) $(\zeros_N, \ud{K,S})$ (Lemma \ref{lem:abstract-soundness}). From the first claim, we can conclude that every $\eps$-\WNE of game $(R'_S,C'_S)$ is essentially an $\eps$-\WNE of game $(R_S,C_S)$. By the second claim and triangle inequality for L1 distance, $\tau \eps$ close to $\ud{K,S}$. The proof of the second claim above follows from essentially a straightforward calculation and incurs a factor $\tau$ loss in the distance and a quadratic loss in the approximation. 

To prove the first claim, we first note the following ``dichotomy'' in the game: either both players play some of first $N$ pure strategies or neither does. This is immediate because of the presence of off-diagonal block with all payoffs $(-2)$ in the two payoff matrices. Next, we observe that if the players' strategy is supported on some of the first $N$ pure strategies, then, in fact, the probability mass on the first $N$ strategies should be $\Theta(1)$ for each. This is again a direct consequence of the presence of the $(-2)$ off-diagnal block in the payoff matrices. 

Next, we observe that if we take any $\eps$-\WNE of $(R'_S, C'_S)$ and condition on using only the first $N$ strategies, then the resulting strategy pair should be an $\eps$-\WNE for $(R,C)$. This is because for any $i\in [N]$, the payoff from playing $i$ given that the second player plays any strategy in the second block is equal and $(-2)$. By the soundness property of $\nice$ game, the  total payoff from this first $N \times N$ block of the game, thus, can be at most $2(\delta -O(\eps))$. On the other hand, for any strategy in the first $N \times N$ block, each of the last $K$ strategies for the other player give a payoff of at least $\delta$ for both players. Thus, there's an incentive for at least one of the players to deviate to playing one of the $K$ strategies. Using the ``dichotomy'' claim above, the other player must also deviate similarly. 

These arguments ultimately yield the following ``stiching-together'' result:

\begin{theorem}\label{thm:abs}
Given parameters $\eps$, $\delta>\eps$, and $\tau>0$, $(i)$ let $(R,C)$ be a degree $t$, $(\eps,\delta)$-\nice game of dimension $N\times N$, and $(ii)$ for an integer $K$, appropriately chosen $\beta$, and subsets $S_1,\dots, S_{q(\beta)}$ of $[K]$ such that $||\ud{K,S_i} - \ud{K,S_j}||_1 > \beta\tau\eps$, let $(R_{S_i}, C_{S_i})$ be an $(\eps,\tau)$-\good game for each $i\in\{1,\dots,q(\beta)\}$.

Then, for the family of games $\CF=\{\CG_i\ |\  \CG_i \mbox{ is the game of \eqref{aeq:abs} constructed using $(R,C)$ and $(R_{S_i},C_{S_i})$}\}$, 

\noindent$(1.)$ All $G_i$'s have a common degree $t$ pseudo-equilibrium. 

\noindent$(2.)$ For any pair of games $\CG_i \neq \CG_j \in \CF$, their $O(\eps^2)$-$\NE$ strategy sets of either players do not intersect. 
\end{theorem}

The above theorem implies that if there exists degree $d$, $q$-query algorithm to find $\eps^2$-$\NE$ in two player games, then $d=\Omega(t)$ and $q=\Omega(q(\beta))$. 

\subsection{$O(n^{\log(n)})$ lower bound for the constant approximation}

Given the discussion in Section \ref{sec:abs}, we only need to give separate constructions of \nice games with degree $d=\Theta(\log(n))$, and $n^{\Omega(\log{(n)})}$-large family of \good games for some constant $\eps>0$. 

\paragraph{Reduction from planted clique does not work, yet} The first candidate here, of course, is the reduction of Hazan-Krauthgamer~\cite{Hazan-Krauthgamer} who show such a reduction from the planted clique problem in average case complexity. A recent work of Barak-Hopkins-Kelner-Kothari-Moitra-Potechin~\cite{DBLP:journals/eccc/BarakHKKMP16} shows a quasi-polynomial SoS lower bound for the planted clique problem. Thus, a priori, it appears that we immediately obtain the \nice game as required here. 

There's a subtle issue, however, that prevents this program from going through. For a reduction to carry over and give a SoS lower bound, we must show that the map that takes a solution to the starting problem (i.e., planted clique) into a solution to the problem of interest (approximate equilibria maximizing social welfare) be computed by a low-degree polynomial (see Fact~\ref{fact:reductions-within-SoS}). 

Now, in the reduction due to Hazan-Krauthgamer, this map is really simple: start from the planted clique, say $S$, and return uniform distribution on pure strategies corresponding to the vertices in the clique, say $\ud{S}$. If $x$ is the indicator vector of the clique, then, the resulting equilibrium is essentially described by $x/|S|$. This appears like a degree 1 map except for one important issue - each non-zero entry of $\ud{S}$ is $1/|S|$. If the size of the clique we planted is ``fixed'', then, this entry is a constant and the map is indeed a degree 1 polynomial. However, technically, fixing the size of the clique requires that the SoS lower bound hold for the feasibility version of the polynomial program for clique that has an explicit size-constraint for the clique. 

Unfortunately, at present, the lower bound in~\cite{DBLP:journals/eccc/BarakHKKMP16} holds only for the optimization version of the clique program that doesn't have an explicit clique size constraint. While we fully expect the SoS lower bound to hold for both versions of the clique program, at present there are significant technical difficulties in extending the lower bound to the feasibility version. 

\paragraph{Relying on known ETH hardness results} The starting point of our construction of the \nice game is the recent work~\cite{Deligkas-Fearnley-Savani} that shows a construction of a game $(R,C)$ via reduction from \emph{free repetitions} of 3SAT (an idea that began in the work of~\cite{DBLP:conf/coco/AaronsonIM14} and first used to show hardness of computing equilibria in~\cite{DBLP:journals/eccc/BravermanKW14}) to show hardnesss of maximizing social welfare in two-player games modulo the ETH. Now, there are standard hardness results for approximating 3SAT via SoS (i.e., Grigoriev's Theorem). Thus, we will be done if we could show that the reduction appearing in the ETH hardness holds also when restricted to SoS algorithm. Long story short, there's a simple sufficient condition to check in this regard (variant of which appeared first in the work of Tulsiani~\cite{DBLP:conf/stoc/Tulsiani09}) and demands only that the map that takes a solution of 3SAT into a solution of the problem at hand be a low-degree polynomial. We give a brief overview of the reduction (see Section \ref{app:reduction} ) in~\cite{Deligkas-Fearnley-Savani} and argue why their reduction induces a low-degree polynomial map between solutions. Checking that the appropriate parameters hold gives us our \nice game as required. 

Such a result might be of independent interest so we state it here in a standalone form. 
\begin{lemma}[\nice game for $O(1)$-$\NE$]
There exists a game $(R,C)$ for $R,C \in [-1,1]^{N \times N}$ such that:
\begin{enumerate}
\item 
 \textbf{Completeness: } There's a degree $\Omega(\log{(N)})$, pseudo-equilibrium with payoffs $\ge 1$ for both players.
\item 
 \textbf{ Soundness: } For every $\eps$-$\NE$ $(\xx,\yy)$ payoffs of both players is at most  $(1-\eps)$ for $\eps \geq 1/1200.$
\end{enumerate}
\label{lem:const1}
\end{lemma}
The starting point of our construction of \good game is the work of~\cite{DP} discussed before, which we briefly note the important aspects of it before proceeding.

Let $m = {l \choose l/2}$ for some even integer $l$ and fix any $\epsilon > 0$. The family of games $(A_S,B_S)$ is indexed by a collections of subsets $S \subseteq [m]$, each of size $\ell$. The key property of this family is that for every $\eps$-\WNE of $(A_S,B_S)$, the support of the row player's strategy is contained in $S$. Further, the row player's strategy is at most $O(\eps)$ L1 distance away from the $\ud{m,S}$ in the $l_1$ norm. Finally, for a given $\beta>0$ there exists a family of at least $n^{(0.8 - 2\beta\eps)\log(m)}$ many such subsets $S$, such that for any pair of subsets $S,S'$, $||\ud{m,S} - \ud{m,S'}||_1 >\beta\eps$. 

As noted before, this construction doesn't give a lower bound against the partial verification oracle as there's an $\epsilon$-\WNE for every game that share the column player's strategy. To construct an \good game, we need such guarantees for both the players. 

It appears natural to somehow take an additional flipped copy of the game with the two players' role reversed and combine them somehow to get $\good$ game. The challenge of course is to combine them in such a way that ensures that all constant approximate equilibria are still close to the appropriate uniform strategy. With a slightly technical argument, we show that there's a simple construction that satisfies such a requirement (see Section \ref{sec:const-approx-good}).

\subsection{Exponential lower bound for the inverse-polynomial approximation} 

As before, we need to construct a degree $\Omega(n)$ \nice game and an exponentially large family of \good games with disjoint set of $O(1/n^2)$-\WNE. 

For a \nice game, as before, we can work with any reduction that $(i)$ starts from a problem that is hard for SoS and $(ii)$ produces a game where approximating the payoff of the best $1/n^2$-approximate \WNE is hard. 

Our idea for constructing the \nice game is again based on a reduction from a problem with known, exponential SoS lower bounds. Specifically, we give a reduction from the independent-set problem to a gap version of approximate Nash equilibrium. Given a graph $G=(V,E)$ on $n$ vertices and a parameter $k$, we construct a $(2n+1) \times (2n+1)$ game $(R,C)$. We show that: $(i)$ if $G$ has a $k$ sized independent set then game $(R,C)$ has an equilibrium with payoff $(1+\nfrac{1}{2k})$ (Lemma \ref{lem:is1}), $(ii)$ if $G$ does not have an independent set of size $k$ then all $(1/5k)$-\WNE of game $(R,C)$ gives payoff at most $1$ (Lemma \ref{lem:is3}).

Our construction is inspired by the classical work of Gilboa and Zemel~\cite{GZ} (see \eqref{eq:exp-nicegame}).

The completeness of this construction is easy. Given an independent set of $G$, it can be argued that uniform distribution on the support of the independent set is in fact an equilibrium for the game above. Further, the size of the independent set, $k$, is fixed as a parameter (equivalently, appears as a constraint in the polynomial program we consider). This implies that the equilibrium is computed via a degree 1 (linear) map and we can immediately apply the technology for importing reductions into SoS framework (Fact~\ref{fact:reductions-within-SoS}). Combining with the known SoS lower bounds for finding independent sets of fixed $k = \Theta(n)$ size~\cite{DBLP:conf/stoc/Tulsiani09} completes the proof.

To construct an \good game, we use generalized matching pennies (GMP) game. 
Here the row-player wants to match strategies with the column-player, while column-player does not want to match, {\em i.e.,} game $(R,C)$ such that $R(i,i)=1, C(i,i)=-1$, and the rest are zero. 
It is easy to see that the uniform distribution over all strategies is the only NE strategy for both the players. We show that in an $m \times m$ game, even $(1/m^2)$-\WNE is at distance at most $O(1/m)$ from the uniform strategy in $l_1$ norm (Lemma \ref{lem:MPgame}). 

Now for every subset $S \subseteq [n]$ we construct an $n \times n$ game $(R_S,C_S)$, and embed $m=|S|$ sized GMP into its $S\times S$ block. 
We make sure that all strategies outside $S$ are strictly dominated by those inside $S$ by a large margin in payoffs. Thus $\eps$-\WNE of game $(R_S,C_S)$ are essentially $\eps$-\WNE of the GMP game, and therefore they concentrate around uniform distribution on $S$, namely $\ud{n,S}$. This gives an \good game (Lemma \ref{lem:exp-goodgame}). 
Finally, the property that $||\ud{n,S} -\ud{n,S'}||_1 \ge \frac{1}{n}$ for any two distinct subsets $S, S'$ of $[n]$, gives us the desired property that the $\eps$-\WNE sets of corresponding \good games $(R_S,C_S)$ and $(R_{S'},C_{S'})$ do not intersect in either player's strategy.
This gives us a family of $2^{\Omega(n)}$ many \good games, corresponding to subsets of $[n]$.

The above construction of degree $\Omega(n)$ \nice game, and a family of $2^{\Omega(n)}$ many \good games, together with Theorem \ref{thm:abs} gives Theorem \ref{thm:exp-lower-bound-technical}.

\section{Stitching \good and \nice games together}\label{app:abs}
In this section, we show how to combine hardness against SoS and hardness against "enumeration" to obtain lower bounds against $\OV$ algorithms. 
For this we will first work with a stronger notion of approximation, namely $\epsilon$-well-supported NE (\WNE), described next.

\subsection{Well-Supported NE vs. Approximate NE}
Recall Definition \ref{def:wne} of $\epsilon$-approximate well-supported Nash equilibrium. We state the condition formally below for the reader: 
For a two player game $(R,C)$, a strategy profile $(x,y)$ is an $\epsilon$-\WNE iff 
\[
\mbox{for every } i \leq n,\ x_i > 0 \Rightarrow e_i^{\top} Ry \ge \max_k e_k^{\top} R y -\epsilon  \ \ \ \mbox{ and } \ \ \ y_i>0 \Rightarrow x^{\top} C e_i \ge \max_k x^{\top} C e_k-\epsilon.
\]

Strategies giving maximum payoff to a player are called her best response, and those giving at least maximum minus $\eps$ payoff $\eps$-best response, defined below:
\[ \begin{array}{ccc} \BR_1(\yy) = \{i\in [n]\ |\  i \in \argmax_k e_k^{\top} R y\} & \mbox{and} & \BR_2(\xx) = \{i\in [n]\ |\  i \in \argmax_k x^{\top} C e_i\} \\
\eps$-$\BR_1(\yy) = \{ i \in [n]\ |\ e_i^\top R\yy \ge \max_k e_k^\top R \yy -\eps\} & \mbox{and} & \eps$-$\BR_2(\xx) = \{ i \in [n]\ |\ x^\top C e_i \ge  \max_k x^\top C e_k -\eps\}
\end{array}
\]

Note that by definition, $(\xx,\yy)$ is an $\eps$-\WNE iff $\Supp(\xx) \subseteq \eBR_1(\yy)$ and $\Supp(\yy) \subseteq\eBR_2(\xx)$. 
The next lemma characterizes NE, approximate NE, \WNE, and relates them. 
\begin{lemma}\label{lem:ne-prop}
For any two player game $(R,C)$ with payoffs in $[0,1]$ the following holds. 
\begin{enumerate}
\item Strategy profile $(\xx,\yy)$ is an NE iff $\Supp(\xx) \subseteq \BR_1(\yy)  \ \ \ \mbox{ and } \ \ \ \Supp(\yy) \subseteq \BR_2(\xx)$.
\item Every $\eps$-\WNE is $\eps$-$\NE$.
\item If $(\xx,\yy)$ is an $\eps$-$\NE$ then $\sum_{i \notin \sqrt{\eps}\mbox{-}\BR_1(\yy)} x_i \le \sqrt{\eps}$ and $\sum_{i \notin \sqrt{\eps}\mbox{-}\BR_2(\xx)} y_i \le \sqrt{\eps}$. 
\item Given $\eps$-$\NE$ $(\xx,\yy)$ we can construct $(3\sqrt{\eps})$-\WNE $(\xx',\yy')$ such that $\Supp(\xx') \subseteq \sqrt{\eps}\mbox{-}\BR_1(\yy)$ and $\Supp(\yy') \subseteq \sqrt{\eps}\mbox{-}\BR_2(\xx)$. 
\end{enumerate}
\end{lemma}
\begin{proof}
Proofs of first and second part follow by definition. The argument for the third and fourth part described next are essentially from Lemma 2.2 of \cite{CDT} due to Chen, Deng, and Teng.

For the third part, let $i^* \in \BR_1(\yy)$. Then, for each $i\notin \sqrt{\eps}\mbox{-}\BR_1(\yy)$ we have $e_i^\top R\yy < e_{i^*}^\top R\yy -\sqrt{\eps}$. Therefore, if $\sum_{i \notin \sqrt{\eps}\mbox{-}\BR_1(\yy)} x_i > \sqrt{\eps}$, then row player's payoff $\xx^{\top}R\yy < \et{i^*} R\yy -\eps$, a contradiction. Similarly for $\yy$.

For the fourth part, construct of $(\xx',\yy')$ from $(\xx,\yy)$ by removing the probability mass from strategies outside $\sqrt{\eps}\mbox{-}\BR_1(\yy)$ and uniformly distribute them on strategies of $\sqrt{\eps}\mbox{-}\BR_1(\yy)$. Since payoffs are in $[0, 1]$, this will increase or decrease payoffs of any strategy at most by $\sqrt{\eps}$. Therefore, $(\xx',\yy')$ is $(3\sqrt{\eps})$-\WNE. And $\Supp(\xx') \subseteq \sqrt{\eps}\mbox{-}\BR_1(\yy)$ and $\Supp(\yy') \subseteq \sqrt{\eps}\mbox{-}\BR_2(\xx)$.
\end{proof}

\subsection{Abstract Game Construction}

Our construction of the hard games are based on combining construction of integrality gap against Sum-of-Squares relaxations computing equilibria that maximize social welfare and games that are hard for enumeration based algorithms. We abstract out the properties of these two kinds of hard games in this section. Our hardness results in the next two sections will invoke the general combining strategy from this section.

First, we define two player games such that maximizing social welfare is hard for SoS.
\begin{definition}[Games Hard for SoS] \label{def:nicegame}
A game $(R,C)$ for $R,C \in [-1,1]^{N \times N}$ is said to be degree $t$, $(\epsilon,\delta)$-\nice for $\eps,\delta\in[0,1]$ if:
\begin{enumerate}
\item \textbf{Completeness: } There's a degree $t$ pseudo-equilibrium for the game $(R,C)$ that has payoffs for both players at least $\delta.$
\item \textbf{ Soundness: } Every $\eps$-\WNE $(\xx,\yy)$ of $(R,C)$ has total payoff $\xx^{\top}(R+C)\yy < 2(\delta-\eps)$.
\end{enumerate}   
\end{definition}

Next we need a game to construct the family of games with disjoint set of equilibria, and the property we need from these is as follows.

\begin{definition}[Games Hard for Enumeration]
\label{def:goodgame}
A $K\times K$ game $(R_S,C_S)$, where $S\subseteq[K]$, is said to be $(\eps,\tau,c)$-\good for $\tau>0$ and  $\eps,c\in[0,1]$ such that $cK$ is an integer, if $R_S,C_S \in [-1, 1]^{K\times K}$, and 
\begin{enumerate}
\item For any strategy $\yy$ of the column player $\exists i \in [K]$ such that $\et{i}R_S\yy \ge 1/2$, and for any $\xx$ of the row player $\exists j \in [K]$ such that $\xx C_S e_j \ge 1/2$. 
\item For any $\eps$-\WNE $(\xx,\yy)$, let $\txx=(x_1,\dots,x_{cK})$ and $\tyy=(y_1,\dots,y_{cK})$. The following holds:
\begin{enumerate}
\item $\Supp(\txx), \Supp(\tyy) \subseteq S$. 
\item $||\txx - c\ud{cK,S}||_1 \le \tau\eps$ and $||\tyy - c\ud{cK,S}||_1\le \tau\eps$. 
\end{enumerate}
\end{enumerate}
\end{definition}

Parameter $c$ in the above definition is needed for the quasi-polynomial lower bound. It is to essentially imply that approximate $\NE$ strategy sets, while projected on first $cK$ coordinates, are disjoint and therefore they are disjoint before projection as well.

Given parameters $1/2 \geq \delta > \eps>0$, $c\in[0,1]$, and $\tau>0$, let $(R,C)$ be an $N\times N$ game that is degree $t$, $(\eps,\delta)$-\nice. And for $K=\poly(n)$ and $S\subset [K]$ let $(R_S,C_S)$ be an $(\eps,\tau,c)$-\good game of size $K\times K$. Construct game $(R',C')$ of size $(N+K)\times (N+K)$ as follows: 

\begin{equation}\label{eq:absGame}
R' = \left[\begin{array}{cc} R \ & \ -2 \\ \delta \ &\ R_S \end{array}\right]\ \ \ \mbox{ and } C' =\left[ \begin{array}{cc} C \ & \ \delta \\ -2 \ & \ C_S \end{array}\right]
\end{equation}

where the constants, namely $\delta$ and $(-2)$, represent matrices of appropriate size with all entries set to the stated constant.

\begin{lemma}[Completeness]
There's degree $t$ pseudo-equilibrium in the game $(R',C')$ that gives a payoff of $1$ to both the players.
\label{lem:abstract-completeness} 
\end{lemma}
\begin{proof}
By construction of $(R',C')$ together with property $(1)$ of a degree $t$, $(\eps, \delta)$-\nice game, we know that there's a pseudo-distribution of degree $t$ on strategy profiles $(\xx,\yy)$ that form an $\NE$ for $(R,C).$ This immediately yields a pseudo-distribution $\pE$ on strategy profiles for the game $(R',C')$ by padding $\xx$ and $\yy$ with zeros on the last $K$ coordinates. We claim that this is a pseudo-distribution on $\NE$ of $(R',C').$ 

Since $\pE$ is a pseudo-distribution on $\NE$ of $(R,C)$, it satisfies the inequality constraints $\xx^{\top} R' \yy \geq e_i^{\top} R' \yy$ for every $i \leq N.$ Further, by property $(1)$, the $\pE$ satisfies the constraint $\xx^{\top} R' \yy \geq \delta$. On the other hand, by construction of $(R',C')$, $e_i^{\top} R' \yy \leq \delta$ since $\yy$ satisfies $\yy_i = 0$ for all $i \geq (N+1).$ A similar argument shows that $\pE$ satisfies the inequality constraint $\xx^{\top} C' \yy \geq \xx^{\top} C' e_i$ for every $i \leq (N+K).$ This yields that $\pE$ satisfies the constraints of $\NE$ for $(R',C')$ as required. 
\end{proof}

To prove the soundness of the construction, we first need to show that every \WNE of game $(R',C')$ is supported on the strategies of the \good game.

\begin{lemma}\label{lem:abstract}

For $0<\eps'<2\eps/3$, every $\eps'$-\WNE $(\xx,\yy)$ of game $(R',C')$, we have $\Supp(\xx), \Supp(\yy) \subset\{(N+1),\dots,(N+K)\}$. 
\end{lemma}
\begin{proof}
Let us denote $x_L = \sum_{i\in[N]} x_i$, $x_R=(1-x_L)$, $y_L=\sum_{i\in[N]} y_i$ and $y_R=(1-Y_L)$. It suffices to show that $x_L=y_L=0$. Suppose not, and $x_L>0$ or $y_L>0$. First we show that $x_L>0$ if and only if $y_L>0$. 
By construction, if $y_L=0$ then payoff of the row player from any $i\in[N]$ is $(-2)$ while from any $i>N$ it is $\delta>0$. In that case if $x_i$ is non-zero for an $i\in [N]$ then it contradicts $(\xx,\yy)$ being $\eps$-\WNE. Therefore, $x_L=0$. Similarly, $x_L=0\Rightarrow y_L=0$ follows.

Next we show that if $x_L>0$ then it should be significantly bigger.
\begin{claim}\label{clm:1}
If $x_L>0$ then $x_L>2/3$. And if $y_L>0$ then $y_L>2/3$. 
\end{claim}
\begin{proof}
Let us prove the second part and the first part follows similarly. 
By property $(1)$ of \good-game  $(R_S,C_S)$ there exists a $j\in\{(N+1),\dots, (N+K)\}$ such that payoff of the row player from strategy $j$, namely $\et{j}R'\yy$, is at least $\delta y_L + \frac{y_R}{2}$. Therefore, if $x_L>0$ the strategy $i\in[N]$ with $x_i>0$ should have payoff at least $\delta y_L + \frac{y_R}{2} - \eps'$. On the other hand, payoff from any $i\le N$ is at most $y_L - 2y_R$, implying. 
\[
\begin{array}{lcl}
y_L - 2 y_R \ge \delta y_L + \frac{1}{2}y_R - \eps' & \Rightarrow & y_L (1-\delta) \ge 2.5 y_R -\eps \ \ \ (\because\ \eps>\eps') \\
& \Rightarrow & y_L(1-\delta) \ge 2.5  - 2.5y_L -\delta \ \ \ (\because \ \delta\ge \eps)\\
& \Rightarrow & y_L(3.5-\delta) \ge 2.5-\delta \Rightarrow y_L \ge 2/3 \ \ \ (\because\ \delta\le 1/2)
\end{array}\qedhere
\]
\end{proof}

Let $\hxx$ and $\hyy$ be truncation of $\xx$ and $\yy$ to first $N$ coordinates, and let $\hxx'=\hxx/x_L$ and $\hyy'=\hyy/y_L$ be their normalized vectors. Then, row player's payoff from $i \le N$ is $\et{i}R\hyy -2y_R$. Furthermore, if $x_i>0$, then this payoff must be at least $\max_{k\le N} \et{k}R\hyy - 2y_R - \eps' \ge \max_{k\le N} \et{k}R\hyy - 2y_R -2\eps/3$. This implies, $\forall i\in[N],\ x_i>0 \Rightarrow \et{i}R\hyy/y_L \ge \frac{1}{y_L}(\max_{k\le N} \et{k}R\hyy - 2\eps/3) \Rightarrow \et{i}R\hyy' \ge \max_{k\le N} \et{k}R\hyy' - \eps$. The last implication uses $y_L>2/3$ from Claim \ref{clm:1}. Similarly, we can show that if $\forall i\in[N],\ y_i>0\Rightarrow \hxx'^{\top}C e_i \ge \max_{k\in[N} \hxx'^{\top}C e_i -\eps$. These two together implies that $(\hxx',\hyy')$ is an $\eps$-\WNE of game $(R,C)$. Therefore, by property $(2)$ of \nice games, $\hxx'^{\top}(R+C)\hyy'\le 2(\delta - \eps)\Rightarrow \hxx^{\top}(R+C)\hyy \le (\delta - \eps) x_L y_L$. 

We have either $\hxx^{\top}R \hyy< (\delta-\eps) x_L y_L$ or $\hxx^{\top}C \hyy< (1-\epshat) x_L y_L$. We will show contradiction for the former and the contradiction for the latter follows similarly by symmetry of the construction of game $(R', C')$.

Suppose, $\hxx^{\top}R \hyy< (\delta-\eps) x_L y_L$. 
There exists $i\in [N]$ such that $x_i>0$ and $\et{i}R\hyy < (\delta-\eps) y_L$. Therefore, $\et{i}R'\yy = \et{i}R\hyy - 2y_R < (\delta-\eps)y_L - 2y_R$. 
By property $(1)$ of \good game $(R_S,C_S)$, exists $j > N$ such that $\et{j}R'\yy \ge \delta y_L + \frac{y_R}{2}$. By requirement of $\eps'$-\WNE, payoff from strategy $i$ above should be at least payoff from $j$ minus $\eps'$. This gives,

\[
(\delta-\eps) y_L - 2y_R > \delta y_L+\frac{y_R}{2} - \eps' > \delta y_L+\frac{y_R}{2} - \eps \Rightarrow (1-y_L) \eps> 2.5 y_R \Rightarrow \eps > 2.5 
\]
A contradiction.

\end{proof}

\begin{lemma}[Soundness]
For $0\le \epshat<\eps/15$ every $\epshat^2$-$\NE$ $(\xx,\yy)$ of game $(R',C')$ satisfies $||\txx - (\zeros_{N}, c\ud{cK,S})||_1, ||\tyy - (\zeros_{N}, c\ud{cK,S})||_1 \le 9(\tau+1)\epshat$, where $\txx$ and $\tyy$ are projection of $\xx$ and $\yy$ respectively on first $(N+cK)$ coordinates, and $\zeros_N$ is an $N$ dimensional zero vector.
\label{lem:abstract-soundness}
\end{lemma}
\begin{proof}

Note that additive scaling of payoff matrices preserves approximate NE, while multiplicative scaling scales the approximation factor by the same amount. Since every entry in $R', C'$ are in $[-2, 1]$ we first add $2$ to all of them and then divide them by $3$. Thus, $(\xx,\yy)$ is now $\epshat^2/3$-NE of the scaled game, which is also its $\epshat^2$-NE. Lets denote the scaled game by $(\rh,\ch)$. 

Then by Lemma~\ref{lem:ne-prop} property $(4)$ there exists $(\xx',\yy')$ that is $3\epshat$-\WNE of game $(\rh,\ch)$, which is $9\epshat$-\WNE of game $(R',C')$. Note that $9\epshat < 2\eps/3$. Applying Lemma \ref{lem:abstract} we get that $\Supp(\xx') , \Supp(\yy') \subseteq\{(N+1),\dots, (N+K)\}$. Therefore, strategy $(\xx',\yy')$ projected to last $K$ coordinates is a $9\epshat$-\WNE of game $(R_S,C_S)$. The latter is an $(\eps,\tau,c)$-\good game, whose property $(2.a)$ implies $||\txx' - (\zeros_N, c\ud{cK,S})||_1 \le 9\tau \epshat$, where $\txx'$ is the projection of $\txx'$ on first $(N+cK)$ coordinates. 

If we get an upper bound on $l_1$ distance between $\xx$ and $\xx'$, then using triangle inequality we will get an upper bound on $l_1$ distance between $\txx$ and $c\ud{cK,S}$. By construction of $(\xx',\yy')$, we have that $\Supp(\xx') \subseteq \epshat$-$\BR_1(\yy)$ and $\Supp(\yy') \subseteq \epshat$-$\BR_2(\xx)$. Furthermore, Part $(3)$ of Lemma \ref{lem:ne-prop} implies $\sum_{i \notin \ehBR_1(\yy)} x_i \le \epshat$ and $\sum_{i\notin\ehBR_2(\xx)} y_i \le \epshat$. While constructing $\xx'$ from $\xx$ we set the probability of all strategies not in $\ehBR_1(\yy)$ to zero and redistribute their mass to strategies in $\ehBR_1(\yy)$ uniformly. Therefore, $||\xx - \xx'||_1 \le 2\epshat$. Finally, $||\txx - (\zeros_N, c\ud{cK,S})||_1 \le ||\xx - \xx'||_1 + ||\txx' - (\zeros_N, c\ud{cK,S})||_1 \le (9\tau+2)\epshat$. 
\end{proof}

Using the \nice game and a family of \good games on far-apart sets, next we construct a family of hard games. The proof of the next theorem will crucially use Lemmas \ref{lem:abstract-soundness} and \ref{lem:abstract-completeness}.

\begin{theorem}\label{thm:absFamily}
Given parameters $\eps\in [0,1/2)$, $\delta>\eps$, $\tau>0$, and $c\in(0,1]$ 
\begin{itemize}
\item[$(h_1)$] let $(R,C)$ be a degree $t$, $(\eps,\delta)$-\nice game of dimension $N\times N$. 
\item[$(h_2)$] and for $K=\poly(N)$ and subsets $S_1,\dots, S_f$ of $[K]$ such that $||\ud{K,S_i} - \ud{K,S_j}||_1 > \frac{18}{c}(\tau+1)\eps$, let $(R_{S_i}, C_{S_i})$ is a $(\eps,\tau,c)$-\good game for each $i\in\{1,\dots,f\}$. 
\end{itemize}

Then, for the game family $\CF=\{G_i\ |\  G_i\mbox{ is the game of \eqref{eq:absGame} constructed using $(R,C)$ and $(R_{S_i},C_{S_i})$}\}$, 
\begin{enumerate}
\item All $G_i\in \CF$ share a common pseudo-equilibrium of degree $t$.
\item 
For any pair of games $G_i \neq G_j \in \CF$, their set of $\eps'$-$\NE$ strategies of either players do not intersect, for all $0\le \eps' \le (\eps/15)^2$.
\end{enumerate}
\end{theorem}

\begin{proof}
The first part for each game $G_i$ follows from the first part of Lemma \ref{lem:abstract-completeness}. 
For the second part, let $(\xx^i,\yy^i)$ and $(\xx^j,\yy^j)$ be $\eps'$-$\NE$ of games $G_i$ and $G_j$ respectively. If we think of $\eps' = \epshat^2$ then clearly $\epshat <\eps/15$ is satisfied. Let $\txx^i, \tyy^i, \txx^j, \tyy^j$ be projection of respectively $\xx^i,\yy^i,\xx^j,\yy^j$ on first $(N+cK)$ coordinates. Using Lemma \ref{lem:abstract-soundness} we have $||\txx^i - (\zeros_N, c\ud{cK,S_i})||_1 \le 9(\tau+1)$ and $||\txx^j - (\zeros_N, c\ud{cK,S_j})||_1 \le 9(\tau+1)$. If $\xx^i=\xx^j$ then $\txx^i=\txx^j$, and thereby using triangle inequality of $l_1$ norm we get $||\ud{cK,S_i} - \ud{cK,S_j}||_1 \le \frac{18}{c}(\tau+1)\eps$, a contradiction to hypothesis $(h_2)$. The same contradiction follows if $\yy^i=\yy^j$.
\end{proof}

\section{Exponential Lower Bounds for $\NE$}
\label{app:exact}
In order to prove hardness against using SoS relaxation together with enumeration, we need to construct a family of games all of which share a common 
high degree pseudo-expectation for its $\eps$-$\NE$ formulation \ref{eq:approx-NEchar}, while their $\eps$-$\NE$ strategy sets for either player are disjoint.
In this section, we obtain exponential hardness for $1/\poly(n)$-$\NE$, and prove Theorem \ref{thm:main-exact}. It follows from the following result. 

\begin{theorem} [Exponential Lower Bound for $\NE$]
For every $n$ large enough, there's a family of $\Gamma = 2^{\Omega(n)}$ two-player games $\{G_i=(R_i,C_i)\}_{i = 1}^{\Gamma}$ with \alert{$O(n)$} pure strategies for both players and all payoffs in $[-1,1]$ such that:
\begin{enumerate}
\item \textbf{Completeness:} All $G_i$s share a degree $\Theta(n)$, pseudo-equilibrium. 
\item \textbf{Soundness:} For any $i\neq j$ if $(\xx,\yy)$ and $(\xx',\yy')$ are \alert{$\Theta(1/n^4)$-$\NE$} of games $G_i$ and $G_j$ respectively then $\xx\neq \xx'$ and $\yy \neq \yy'$. 
\end{enumerate}
\label{thm:exp-lower-bound-technical}
\end{theorem}


The proof of Theorem \ref{thm:exp-lower-bound-technical} follows immediately from the following two constructions of hard games along with the generic combining construction from Theorem \ref{thm:absFamily}.

\begin{lemma}[Games Hard for SoS] \label{lem:exp-sos-hard}
There exists a degree $\Theta(n)$, $(\epsilon,\delta)$-\nice game $(R,C)$ for $\eps=O(1/n^2)$, $\delta\in(0,1/2]$ and $R,C \in [-1,1]^{O(n) \times O(n)}$.   
\end{lemma}

\begin{lemma}[Games Hard for Enumeration]\label{lem:exp-enum-hard}
For $K=O(n)$, $\eps=O(1/n^2)$, $\tau=O(n)$ and $c=1$, there exists a family of subsets $\CF\subset 2^{\{0,1\}^K}$ of size at least $2^{\Omega(n)}$ such that 
\begin{itemize}
\item for each $S\in \CF$ there is an $(\eps,\tau,c)$-\good game $(R_S,C_S)$ for $R_S,C_S \in [-1,1]^{K\times K}$.
\item for any pair of distinct $S,S' \in \CF$, $\norm{\ud{cK,S} - \ud{cK,S'}} > \frac{18}{c}(\tau+1)\eps$. 
\end{itemize}
\end{lemma}

\subsection{\nice Game Construction}\label{sec:expnice}
\alert{In Lemma \ref{lem:exp-sos-hard}, the \nice game has to encode degree $\Theta(n)$ pseudo-equilibrium with high payoff, while no such approximate $\NE$ exists. Among the combinatorial problems, such SoS lower bound is known for independent set. If we construct a game that encodes independent set in a NE with high payoffs such that map from the independent set to corresponding NE is low degree then we are done using Fact \ref{fact:reductions-within-SoS}. This is exactly what we achieve in this section.}

Given an undirected graph $G$ on $n$ vertices represented by an adjacency matrix $E$ (abuse of notation), we wish to check if it has an independent set of size $k\le n$ ($k$-IS). We will reduce this problem first to finding a Nash equilibrium in a two-player game with certain property. In fact we will show that the game has a NE with payoff at least $\delta=(1+1/2k)$ if $k$-IS has a solution, otherwise all it's $\eps$-\WNE have payoff at most $(\delta-\eps)$ for an inverse polynomial $\eps$. 

For $\gamma=1/2$, $A=\ones_{n\times n} -E + \gamma I_{n\times n}$, $B=-M* I_{n\times n}$ and $B'=(k+\gamma)*I_{n\times n}$, consider the following $(2n+1) \times (2n+1)$ block matrices:
\begin{equation}\label{eq:exp-nicegame}
R=\left[\begin{array}{cc}
\begin{array}{cc}
A & B\\
B' & \zeros_{n\times n}\\
\end{array}
& (-1) \ones_{2n \times 1}\\
(1+\frac{\gamma}{k}) \ones_{1\times 2n} & 1
\end{array}\right]
\ \ \ \ \ \ 
C=\left[\begin{array}{cc}
\begin{array}{cc}
A & B'\\
B & \zeros_{n\times n}\\
\end{array}
& (1+\frac{\gamma}{k}) \ones_{2n\times 1}\\
(-1) \ones_{2n\times 1}& 1
\end{array}\right]
\end{equation}

In the above construction, submatrix $A$ has entry one in $(u,v)$th position if edge $(u,v)$ is {\em not present} in graph $G$, otherwise the entry is zero. Further, all its diagonals are $(1+\gamma)=1.5$. Note that game $(R,C)$ is symmetric, i.e., $C=R^T$. For ease of notation let $m=2n+1$. 

\begin{lemma}\label{lem:is1}
If graph $G$ has an independent set $S\subseteq V$ of size $k$ then strategy profile $(\xx,\xx)$ where $\xx=(\ud{n,S}, \zeros_{(n+1)})$ is a Nash equilibrium of game $(R,C)$. Furthermore, payoffs of both the players at this NE is $(1+1/2k)$. 
\end{lemma}
\begin{proof}
This is because payoff from each non-zero strategy is $(1+\gamma/k)$, and from strategies with zero probability it is at most $(1+\gamma/k)$. 
\end{proof}

The above lemma constructs a symmetric $\NE$ with high payoff. In what follows, we consider non-existence of both symmetric and non-symmetric NE with high payoff when $k$-sized independent set does not exist in the graph.

\begin{lemma}\label{lem:is3}
Given that graph $G$ does not have an independent set of size $k$, if $(\xx,\yy)$ is an $\eps$-\WNE of game $(R,C)$ for $\eps \le \frac{1}{5k}$, then $\Supp(\xx)=\Supp(\yy)=\{(2n+1)\}$. 
\end{lemma}
\begin{proof}
To the contrary suppose $\Supp(\yy)\neq\{(2n+1)\}$. This proof is based on a number of observations which we list next with a brief justification.

\begin{enumerate}
\item $\Supp(\yy) \cap \{1,\dots,n\} \neq \emptyset$ and $\Supp(\xx) \cap \{1,\dots,n\} \neq \emptyset$: If $\Supp(\yy) \subseteq\{(n+1),\dots, (2n+1)\}$, then $\Supp(\xx)=\{(2n+1)\}$ since any strategy other than the last one gives non-positive payoff to player one, while the last one gives payoff of at least $1$. In that case, by the same argument $\Supp(\yy)=\{(2n+1)\}$, a contradiction. Second part follows similarly. 


\item $\forall i\le n, x_i > 0 \Rightarrow y_i \ge \frac{1}{k} - \frac{\eps}{\gamma}$, and $y_i>0 \Rightarrow x_i \ge \frac{1}{k}-\frac{\eps}{\gamma}$: 
For the first part, suppose $x_i>0$. Row player's payoff from $i$th strategy is $\et{i}R\yy \le 1+\gamma y_i - y_m$. While $m$th strategy gives her $\et{m}R\yy \ge 1+\gamma/k (1-y_m)$. Since $(\xx,\yy)$ is $\eps$-\WNE, the former payoff has to be at least the latter payoff minus $\eps$, 
\[
1+\gamma y_i - y_m \ge 1+ \frac{\gamma}{k} (1-y_m) -\eps \Rightarrow \gamma y_i \ge \frac{\gamma}{k} + y_m (1-\frac{1}{k}) -\eps \Rightarrow y_i \ge \frac{1}{k}- \frac{\eps}{\gamma}
\]

\item $\Supp(\xx) \cap \{1,\dots,n\}= \Supp(\yy) \cap \{1,\dots,n\}$ and $|\Supp(\xx) \cap \{1,\dots,n\} \le (k-1)$: From Point $(2)$, $\Supp(\xx) \cap \{1,\dots,n\}= \Supp(\yy) \cap \{1,\dots,n\}$ follows. For the second part, to the contrary suppose $|\Supp(\xx) \cap \{1,\dots,n\}|\ge k$. Since $G$ has no independent set of size $k$, there exists $u\neq v \in \Supp(\xx) \cap \{1,\dots, n\}$ such that edge $(u,v)$ is present in $G$ and in turn $A(u,v) =0$. Sum of the payoffs of the row player from these two strategies are:
\[
\et{u}R\yy + \et{v}R\yy \le (2 - y_u - y_v) + \gamma (y_u + y_v) = 2 - \frac{(y_u + y_v)}{2} \le 2 - \frac{1}{k} + 2\eps
\]
The last inequality uses the fact that from Point $(2)$ that if $x_u,x_v > 0 $ then $y_u,y_v \ge \frac{1}{k} - \frac{\eps}{\gamma}$ and $\gamma=1/2$. Thus either strategy $u$ or strategy $v$ gives payoff at most $(1 - \frac{1}{2k} + \eps)$. While the payoff from $m$th strategy is at least $1$. Then, $\eps$-\WNE ensures that $(1 - \frac{1}{2k} + \eps) \ge 1-\eps \Rightarrow \eps \ge \frac{1}{4k}$, a contradiction to $\eps \le \frac{1}{5k}$. 
\end{enumerate}

Let $y_L=\sum_{i\in[n]} y_i$ and $a=|\Supp(\yy) \cap \{1,\dots,n\}|$ Then from Point $(1)$ we have $y_L>0$ and from Point $(3)$ we have $a\le (k-1)$. There exists strategies $i,i' \in \Supp(\yy) \cap \{1,\dots,n\}$ such that $y_i \le y_L/a$ and $y_{i'}\ge y_L/a$. Again using Point $(3)$ we have $x_i>0$ as well, and the corresponding payoff is $\et{i}R\yy \le y_L + \gamma y_i -  y_m \le y_L + \gamma y_i\le (1+\gamma /a)y_L$, while the row player's payoff from $(n+i')^{th}$ strategy is $\et{(n+i')}R\yy = (k+\gamma) y_{i'} \ge (k+\gamma)y_L/a$. Therefore,
\[
\et{(n+i')}R\yy - \et{i}R\yy \ge (k/a + \gamma/a - 1 - \gamma/a) y_L = (k - a) y_L/a \ge y_L/a \ge \frac{a(\min_{i\in[n]} y_i)}{a} \ge \frac{1}{k} - 2\eps > \eps. 
\]
The second last inequality follows using point $(2)$ and $(3)$, while the last follows using $\eps \le \frac{1}{5k}$. This contradicts strategy $i$ giving at least best payoff minus $\eps$, i.e., contradiction to $(\xx,\yy)$ being $\eps$-approximate well-supported equilibrium. 
\end{proof}

Using Lemmas \ref{lem:is1} and \ref{lem:is3} next we show that $(R,C)$ gives an \nice game. 
Note that $R,C\in[-1, k]$.

We rely on a key result that is a variant of the one developed by Tulsiani \cite{DBLP:conf/stoc/Tulsiani09}. We include an elementary proof in the appendix. 

\begin{fact}[\cite{DBLP:conf/stoc/Tulsiani09}] \label{fact:reductions-within-SoS}
Let $\Psi_1(x), \Psi_2(y) = \{g_1 = 0, g_2 = 0, \ldots, g_q = 0, h_1 \geq 0, h_2 \geq 0, \ldots, h_r \geq 0\}$ be two systems of polynomial constraints in $x \in \R^{n_1}$ and $y \in \R^{n_2}$, respectively. 

Suppose there are polynomials $f_1, f_2, \ldots, f_{n_2}$ of degree at most $t$, and sum-of-squares polynomial $Q_1,Q_2, \ldots, Q_r$ such that for every $x$ that is feasible for $\Psi_1$, $y \in \R^{n_2}$ defined by $y_j = f_j(x)$ for every $j \in [n_2]$ satisfies:
\begin{enumerate}
	\item $g_j(y) = 0$ for every $1 \leq j \leq q$, and
	\item $h_j(y) = Q_j(y).$ 
\end{enumerate}

Then, if there exists a degree $d$ pseudo-expectation satisfying the constraints of $\Psi_1$, then there exists a degree $\lfloor d/t \rfloor$-pseudo-distribution satisfying the constraints of $\Psi_2.$
\end{fact}

\begin{lemma}\label{lem:nicegame}
Given the combinatorial problem $k$-IS on graph $G=(V,E)$ on $n$ vertices, construct $(R,C)$ as per (\ref{eq:exp-nicegame}). 
Game $\frac{1}{k}(R, C)$ is a degree $\Omega(n)$, $(\eps,\delta)$-\nice, where $\eps=\frac{1}{5k^2}$ and $\delta=(\frac{1}{k}+\frac{1}{2k^2})$.
\end{lemma}
\begin{proof}
First note that all payoff entries of game $(\tR,\tC)=\frac{1}{k}(R,C)$ are in $[-1,1]$. For the {\em completeness} property of an \nice game, note that 
solution to independent-set gives a Nash equilibrium of $(\tR,\tC)$ with payoff at least $(\frac{1}{k}+\frac{\gamma}{k^2}) = \delta$ (Lemma \ref{lem:is1}). Thus, pseudo-distribution for the independent set can be mapped to one for Nash equilibrium using Fact \ref{fact:reductions-within-SoS} on reductions within SoS framework. Furthermore, if the size of the independent set is fixed, then this map is low degree (Fact \ref{fact:reductions-within-SoS}). Finally, the known degree $\Omega(n)$ SoS hardness for independent set of fixed $O(n)$ size \cite{DBLP:conf/stoc/Tulsiani09}, gives degree $\Omega(n)$ pseudo-equilibrium for game $(\tR,\tC)$ with payoff $\delta$. 

The {\em soundness} property follows using Lemma \ref{lem:is3} since when both players play their last strategy, they both get payoff $\frac{1}{k} < (\delta - \eps)$. 
\end{proof}

Lemma \ref{lem:nicegame}, the fact that $k\le n$, together with the SoS hardness of independent set and Fact \ref{fact:reductions-within-SoS} give Lemma \ref{lem:exp-sos-hard}.


\subsection{\good Games Construction}\label{sec:expgood}
In this section we prove Lemma \ref{lem:exp-enum-hard}. 
For this we construct an $n \times n$ \good game for a given subset $S\subset [n]$. Let $|S|=m$. We will first construct an $m \times m$ game whose $O(1/m^2)$-\WNE set is near the uniform strategy. 

\paragraph{Generalized Matching Pennies \cite{DGP}}
This game was introduced in \cite{DGP}, however, a characterization of its approximate equilibria was not known to the best of our knowledge. Consider a zero-sum two-player game with $m$ strategies, where the row player wants to match with the column player and the column player does not want to match. That is if both plays $i\in[m]$ then the row player one gets payoff $1$ and the column player gets $-1$. Otherwise, both gets zero. If $I(m)$ denotes the identity matrix of size $m$ then, the game essentially is $(I(m), -I(m))$. It is easy to see that both players playing uniform random strategy is the only Nash equilibrium \cite{DGP}. Next we show that all $\eps$-\WNE are also near this unique NE.

\begin{lemma}\label{lem:MPgame}
For $0\le \eps \le \frac{1}{m}$ if $(\xx,\yy)$ is an $\eps$-\WNE of the generalized matching-pennies game $(I(m), -I(m))$ then $||\xx - \ud{m}||_1 \le m\eps$ and $||\yy - \ud{m}||_1 \le m\eps$. 
\end{lemma}
\begin{proof}
We will argue for $\xx$, and similar argument follows for $\yy$.
It suffices to show that $||\xx - \ud{m}||_\infty \le \eps$. To the contrary suppose $\exists i\in[m]$ such that $x_i> 1/m + \eps$. Then $\exists i' \in [m]$ such that $x_{i'} < 1/m$. In that case, $i \notin \eBR_2(\xx)$ and therefore $y_i=0$. This would imply existence of a $j \in [m]$ such that $y_j\ge \frac{1}{(m-1)} > \eps$. Now note that row player's payoff from strategy $j$ is $y_j>\eps$ while from strategy $i$ it is zero. Therefore $i\notin \eBR_1(\yy)$, a contradiction to $(x,y)$ being an $\eps$-\WNE.
\end{proof}

For any given $S\subseteq [n]$ with $|S|=m$, consider the following $n \times n$ game $(R_S,C_S)$. 
\begin{equation}\label{eq:exp-goodgame}
\begin{array}{l}
\mbox{{\bf Dummy strategies:} $\forall i \notin S, \forall j \in [n]$ set $R_S(i,j)=C_S(j,i)=-1$.} \\
\mbox{{\bf Matching pennies:} $\forall i, j \in S,$ set $R_S(i,j)=I(m)(i,j)+1$ and $C_S(i,j)=-I(m)(i,j)+3$.}\\  
\mbox{{\bf Rest:} $\forall i\in S,\ \forall j \notin S$ set $R_S(i,j)=1$ and $C_S(j,i)=2$.}\\
\mbox{{\bf Normalization:} Divide $R_S$ by $2$, and $C_S$ by $3$.}\\
\end{array}
\end{equation}


\begin{lemma}\label{lem:exp-goodgame}
For any given $S\subseteq [n]$, game $(R_S, C_S)$ of (\ref{eq:exp-goodgame}) is $(\eps,\tau,c)$-\good for any $\eps\le \frac{1}{3n}$, $\tau=3n$, and $c=1$.
\end{lemma}
\begin{proof}
By construction payoffs in $(R_S,C_S)$ are in $[-1,1]$, and both players get payoff at least $1/2$ from any strategy $i \in S$ no matter what the opponent is playing. Therefore, property $(1)$ of \good game (Definition \ref{def:goodgame}) is satisfied. 

Let $|S|=m \le n$. 
For property $(2)$, note that no matter what the opponent is playing, payoff from strategy $i \notin S$ is at most $-1/3$ which is way less than the payoff of $1/2$ from any strategy $i\in S$. Therefore, if $(\xx,\yy)$ is an $\eps$-\WNE, then $\Supp(\xx),\Supp(\yy) \subseteq S$. Therefore, strategy profile $(\xx',\yy')$, the projection of $(\xx,\yy)$ on to coordinates of set $S$, gives a $3\eps$-\WNE of the matching pennies game $(I(m), -I(m))$. Thereby using Lemma \ref{lem:MPgame} it follows that $||\xx-\ud{n,S}||_1 = ||\xx'-\ud{m}||_1 \le 3\eps m \le 3\eps n$. Since $||\xx - \ud{n,S}||_1 = ||\xx'-\ud{m}||_1 \le 3\eps n$, and this upper bound is $\tau \eps$ in property $(2)$, $\tau=3n$ suffices. Similarly we can show $||\yy - \ud{n,S}||_1 \le \tau\eps$. 
\end{proof}

Finally to enable construction of a family of \good games, we will need the following lemma. 
\begin{lemma}\label{lem:exp-goodgame2}
For any two sets $S,S' \subseteq [n]$, $S\neq S'$, $||\ud{n,S} - \ud{n,S'}||_1 \ge \frac{1}{n}$. 
\end{lemma}
\begin{proof}
Without loss of generality suppose $|S|\ge |S'|$, then $\exists i \in S\setminus S'$. Therefore, $||\ud{n,S} - \ud{n,S'}||_1 \ge \ud{n,S}(i) \ge \frac{1}{|S|} \ge \frac{1}{n}$. 
\end{proof}

Now we are ready to prove Lemma \ref{lem:exp-enum-hard} using Lemmas \ref{lem:exp-goodgame} and \ref{lem:exp-goodgame2}.
\begin{proof}[Proof of Lemma \ref{lem:exp-enum-hard}]
We will construct a family of $n \times n$ \good games. For every subset $S\subseteq [n]$, construct an $n\times n$ game $(R_S,C_S)$ as per \eqref{eq:exp-goodgame}. Lemma \ref{lem:exp-goodgame} implies that $(R_S,C_S)$ is an $(\eps,\tau,c)$-\good game for $\eps \le \frac{1}{60 n^2}$, $\tau=3n$, and $c=1$. Furthermore, for any two distinct sets $S,S' \subseteq [n]$, Lemma \ref{lem:exp-goodgame2} implies $||\ud{n,S} - \ud{n,S'}||_1 \ge \frac{1}{n} > \frac{18}{c}(\tau+1)\eps$.  
\end{proof}

\section{Quasipolynomial Hardness for Finding Approximate Equilibria}
\label{app:const}
The goal of this section is to prove Theorem~\ref{athm:approx-nash-hardness-technical}. We restate it for convenience.

\begin{theorem} [Quasi-polynomial Hardness for Approximate $\NE$]
For every $n$ large enough, there's a family of $\Gamma = n^{\Omega(\log{(n)})}$ two-player games $\{(R_i,C_i)\}_{i = 1}^{\Gamma}$ with $n$ pure strategies for both players and all payoffs in $[-1,1]$ such that:
\begin{enumerate}
\item \textbf{Completeness:} All $G_i$s have a common degree $\Theta(\log{(n)})$, pseudo-equilibrium. 
\item \textbf{Soundness:} For a fixed $\eps=O(1),$ for any $i,j$ and any pair of $\epsilon$-$\NE$ in $G_i$, $G_j$, say $(\xx,\yy)$ and $(\xx',\yy')$ respectively, $\xx \neq \xx'$ and $\yy\neq \yy'.$ 
\end{enumerate}
\label{thm:approx-nash-hardness-technical-section}
\end{theorem}

The proof follows immediately from the following two constructions of hard games along Theorem \ref{thm:absFamily}.

\begin{lemma}[Games Hard for SoS] \label{lem:sos-hard-const-approximate}
There exists a degree $\Omega(\log{(n)})$, $(\epsilon,1)$-\nice game $(R,C)$ for $\eps=O(1)$ and $R,C \in [-1,1]^{n \times n}$.   
\end{lemma}

\begin{lemma}[Games Hard for Enumeration]\label{lem:enum-hard-const-approximate}
For $K=O(n)$, $\eps=O(1)$, $\tau=O(1)$ and $c=1/2$, there exists a family of subsets $\CF\subset 2^{\{0,1\}^K}$ of size at least $K^{\Omega(\log(K))}$ such that 
\begin{itemize}
\item for each $S\in \CF$ there is an $(\eps,\tau,c)$-\good game $(R_S,C_S)$ for $R_S,C_S \in [-1,1]^{K\times K}$.
\item for any pair of distinct $S,S' \in \CF$, $\norm{\ud{cK,S} - \ud{cK,S'}} > \frac{18}{c}(\tau+1)\eps$. 
\end{itemize}
\end{lemma}

\subsection{Constructing an \nice game}\label{app:reduction}
We will next show construction of \nice game for a constant $\eps>0$. In~\cite{Deligkas-Fearnley-Savani}, the authors give a reduction from 3SAT to finding constant-approximate Nash equilibrium maximizing social welfare up to some additive slack. This reduction maps a $n$ variable, $O(n)$-clause 3SAT instance into a $2^{O(\sqrt{n})}$ size two-player game. This gives a quasi-polynomial lower bound for the latter problem, assuming that $n$-variable, $O(n)$-clause 3SAT instances require exponential time to solve. 

We simply observe that their reduction is low-degree, that is, the map that takes any satisfying assignment of the starting 3SAT instance into an approximate equilibrium maximizing social welfare is a polynomial of degree $O(\sqrt{n})$. Combined with Fact~\ref{fact:reductions-within-SoS} as in the previous section, we immediately obtain the $\nice$ game as required. Our reduction will in fact yield hardness of finding the weaker solution concept of $\epsilon$-\WNE (thus implying the hardness of $\epsilon$-$\NE$).

\begin{lemma}[SoS Hardness for Finding Payoff Maximizing $\epsilon$-$\NE$]
There exists a game $(R,C)$ for $R,C \in [-1,1]^{N \times N}$ such that 
\begin{enumerate}
\item \textbf{Completeness: } There's a degree $\Omega(\log{(N)})$, pseudo-equilibrium for $(R,C)$ that has payoffs for both players at least $1$.
\item \textbf{ Soundness: } For every $\eps$-$\NE$ $(\xx,\yy)$ of $(R,C)$, both row and column players have an individual payoff of at most $(1-\eps)$ for $\eps< 1/1200.$
\end{enumerate}
\end{lemma}
For completeness, we sketch the argument of Deligkas-Fearnley-Savani~\cite{Deligkas-Fearnley-Savani} next. 
The starting point is the hardness of approximation for the 3SAT problem. To show ETH-hardness, one has to start from a (worst-case hard) instance of 3SAT and applies the PCP theorem to get a hard-to-approximate instance. For SoS-hardness, we only need such a result to hold for the SoS algorithm. As a result, we can rely on Grigoriev's theorem to get exponential SoS-hardness of approximating 3SAT. One can think of this as saying that ETH is unconditionally true when restricted to the SoS framework and, in fact, the hard instance is simply a random one!

\begin{fact}[Grigoriev's Theorem] \label{fact:Grigoriev}
Let $\phi$ be a random instance of 3SAT with $m =\Theta(n/\epsilon^2)$ constraints and $n$ variables. Then, with probability at least $1-1/n$, 
\begin{itemize}
\item For every $x \in \zo^n$, the fraction of clauses satisfied in $\phi$ by $x$ is at most $7/8+\epsilon$, and
\item there exists a degree $\Theta(n)$ pseudo-distribution over the satisfying assignments for the instance $\phi.$
\end{itemize}
\end{fact}

Given a hard instance of 3SAT, one transforms it into a 2 prover, \emph{free} game. 

Recall that a two prover game is defined by a distribution $\cD$ over a finite set of question pairs $(Q_1, Q_2)$ with a collection of admissible answer pairs and two players Alice and Bob. In a run of the game, the referee draws a question pair $(Q_1, Q_2)$ from $\cD$ and sends $Q_1$ to Alice and $Q_2$ to Bob and receives answers from them and accepts if the answer pair is in the list of admissible answer pairs for $(Q_1, Q_2).$ A two prover game is said to be \emph{free} if $\cD$ is a product distribution.

A \emph{strategy} for a two prover game is a  pair of functions $f_1$ and $f_2$ that maps questions to an answer. The value of the strategy is the probability that  $(f_1(Q_1),f_2(Q_2))$ belongs to an admissible answer pair for the question pair $(Q_1, Q_2)$ picked at random from $\cD.$

\begin{definition}[Free Game from 3SAT] \label{def:free-game}
Let $\phi$ be a 3SAT instance from Fact \ref{fact:Grigoriev} with $m$ clauses and $n$ variables. 
Let the clauses and variables appearing in $\phi$ each be partitioned into $\sqrt{n}$ different groups, say, $S_1, S_2, \ldots, S_{\sqrt{m}}$ and $T_1, T_2, \ldots, T_{\sqrt{n}}$, respectively. Then, $\phi$ defines a free game as follows:

Let $\cD$ be the uniform distribution on $(S_i,T_j)$ for $i \sim [\sqrt{m}], j \sim [\sqrt{n}]$ define the distribution on question pairs. For $(S_i,T_j)$, an admissible answer pair is an assignment to the variables appearing in each clause in $S_i$ and each variable in $T_j$ such that 1) every clause in $S_i$ is satisfied by the assignment chosen 2) for every variable that appears in some clause in $S_i$ and in $T_j$, the assignments given by the answer pairs match.  

Note that the number of strategies for Merlin 1 are at most $N_1 = 2^{3\sqrt{m}}$ and that for Merlin 2, at most $N_2 = 2^{\sqrt{n}}.$
\end{definition}

The free game $\cF$ can be thought of as a cooperative two player game as follows. The row player has as their strategies the set of all ordered pairs of questions to Merlin 1 and admissible answers for that question. Similarly, the column player has as strategies, all possible questions to Merlin 2 and admissible answers for those questions. Finally, a pair of ordered pairs of questions and answers receive a payoff of 1 (for both row and column player) iff the verifier accepts the two answers on those two questions and $0$ otherwise.  Observe that the size of the cooperative game produced from a free game is at most $\sqrt{m} N_1 \times \sqrt{n} N_2.$

The following completeness result holds for every free game produced from 3SAT instances from Fact \ref{fact:Grigoriev}. 
\begin{lemma}[SoS Completeness for Free Game]
Let $\phi$ be the 3SAT instance with $m = \Theta(n/\epsilon^2)$ clauses and $n$ variables from Fact \ref{fact:Grigoriev}. 
Let $\cF$ be the free game constructed in Definition \ref{def:free-game} from $\phi.$  
Then, there's a pseudo-equilibrium of degree $\Theta(n)$ for the cooperative game defined by $\cF$ with payoff for each player $=1$.
\end{lemma}
\begin{proof}
We will start from the pseudo-distribution of degree $\Theta(n)$ given to us by Fact \ref{fact:Grigoriev} that satisfies every constraint of $\phi$ and use Fact \ref{fact:reductions-within-SoS}. To do this, we need to build a low-degree map from any satisfying assignment $x$ of $\phi$ to a strategy for the Merlins in the free game $\cF$. This is very simple: for any question set $(S_i,T_j)$ - the strategy is described by a bit string of length $3 \sqrt{m}$ and $\sqrt{n}$ respectively - where the bit string corresponds to the assignment for the variables appearing in the clauses in $S_i$ and variables in $T_j$ respectively. For the strategy for the cooperative game we use the following mixed strategy: Merlin 1 chooses a question from $X$ uniformly at random and then chooses the answer corresponding to the assignment $x$ for that question. Similarly, Merlin 2 chooses a question from $Y$ uniformly at random and then chooses the answer corresponding to $x$ for that questions. Using $x_i$ as the assignment for the variable $i$ wherever it appears gives a degree 1 map into strategies of $\cF$. 

It is easy to see that since $x$ (purportedly) satisfies all constraints of $\phi$, it satisfies all clauses in $S_i$ and further, the assignments to variables appearing in both $S_i$ and $T_j$ match. Thus, we obtain a degree $1$ map from any satisfying assignment $x$ into a strategy for $\cF$ with value $1$. Further, it is easy to see (since the maximum possible payoff in the game is $1$) that the strategy constructed above satisfies all constraints of being an $\NE$. Thus, we have exhibited a degree $O(\sqrt{n})$ map from any satisfying assignment $x$ to a $\NE$ of the cooperative game associated with $\cF$. By Fact \ref{fact:reductions-within-SoS}, we immediately obtain a pseudo-distribution of degree $\Theta(n)$ on strategies with value $1$ for $\cF$. 

\end{proof}

The following soundness result holds for the partitions $S_1, S_2, \ldots, S_{\sqrt{m}}$ and $T_1, T_2, \ldots, T_{\sqrt{n}}$ chosen appropriately (we invite the reader to find the details of the reduction in Deligkas-Fearnley-Savani~\cite{Deligkas-Fearnley-Savani}):

\begin{fact}[Soundness, \cite{Deligkas-Fearnley-Savani}]
There is no strategy for $\cF$ produced from $\phi$ in Fact \ref{fact:Grigoriev} with value exceeding $1-\Omega(1).$
\end{fact}

Deligkas-Fearnley-Savani give a construction of a game $(R',C')$ by taking the cooperative game $(R,C)$ produced from the free game as above and using it to fill the first $\sqrt{m} N_1 \times \sqrt{n} N_2$ block and then adding $H$ rows and columns for $H$ being $2^{|Y|} \approx 2^{\sqrt{n}}.$ The off diagonal blocks are filled with a carefully constructed zero sum game and its transpose. Finally, the bottom diagonal $H \times H$ block is filled with the all zero matrix. It is immediate from the simple argument above that the pseudo-equilibrium described above for $\cF$ is in fact a $\epsilon$-approximate pseudo-equilibrium in $(R',C')$ - the only required fact in this argument is that the equilibrium strategy $\yy$ ($\xx$ respectively) satisfy $\sum_{b} \yy(y,b) = 1$ for every question $y$ to Merlin $2$. This is a linear polynomial equality constraint that carries over to the pseudo-distribution in a straightforward way. 

We omit the simple and standard argument here and summarize the result that immediately follows below:
\begin{fact}[Corollary of \cite{Deligkas-Fearnley-Savani}]
Let $\cF$ be a free game with $\sqrt{m}$ ($\sqrt{n}$, respectively) possible questions and $N_1$ ($N_2$, respectively) possible strategies for Merlin 1 (2, respectively.)  
Then, for $\eps=O(1)$ there's a two player game with $\sqrt{m} \cdot N_1$ strategies for the row player and $\sqrt{n} \cdot N_2$ strategies for the column payer such that: 
\begin{enumerate}
\item \textbf{Completeness: } There's a $\epsilon$-approximate pseudo-equilibrium of degree $\Theta(n)$ = $\Omega(\log{(N_1 + N_2)})$ where both players receive a payoff of at least $1$. 

\item \textbf{Soundness: } For every $\epsilon$-approximate $\NE$ of the game, the payoff of either player is at most $(1-\eps)$ for $\eps<1/1200.$ 
\end{enumerate}

\end{fact}

\subsection{Constructing a family of \good games}\label{sec:const-approx-good} 
In this section we prove Lemma \ref{lem:enum-hard-const-approximate} by constructing a family of \good games for a constant $\eps$. Essentially, we will construct a family of two-player games, each parameterized by a subset $S$ of strategies such that: $(i)$ both players' strategies at Nash equilibrium is largely supported on the set $S$, and $(ii)$ for a constant $\eps$, all its $\eps$-\WNE are near the actual NE. Daskalakis and Papadimitriou~\cite{DP} gave a construction that ensures row-player's all \WNE strategies are near the uniform distribution over the subset. Constructing such a game for $n^{log(n)}$ carefully selected subsets, they showed non-existence of an oblivious PTAS for \WNE that uses only Verification Oracle \ref{or1}, i.e., checks if given strategy profile is a \WNE. However, in their construction the column player's NE (\WNE) strategies are common across games. Therefore it breaks under our more powerful Verification Oracle \ref{or2}, {\em i.e.,} given any {\em one} player's strategy the oracle checks if there exists a $\eps$-$\NE$ corresponding to the given strategy. 

Since we need hardness against Oracle \ref{or2}, construction of~\cite{DP} can not be used directly. We will have to ensure that, for each player, her NE strategy sets are far apart across games. Our construction builds on that of~\cite{DP}, and therefore first we discuss the latter and its result. 

\paragraph{Construction of \cite{DP}} Let $l$ be an even integer and $n={l \choose l/2}$. For each $S\subset [n]$ of size $l$, we will construct an $n \times n$ game. Fix such a subset $S$, and denote the corresponding game by $(A_S, B_S)$. In this game, row player cares about only strategies from $S$, {\em i.e.,} strategies outside of $S$ are strictly dominated by that of $S$. Since $|S|$ is $l$, number of $(l/2)$ sized subsets of $S$ are exactly $n$.  
The $n$ strategies of the column player correspond to these $(l/2)$ sized subsets of set $S$. In other words, if $S_1,\dots,S_{{l \choose l/2}}$ are all the $(l/2)$ subsets of $S$, then think of column player's strategies being indexed by these subsets. 
The payoffs in column $j$, corresponding to subset $S_j$ of $S$, are:

\begin{itemize}
\item For each $i\notin S$, set $A_S(i,j)=-1$ and $B_S(i,j)=1$.
\item For each $i \in S\cap S_j$, set $A_S(i,j)=1$ and $B_S(i,j)=0$.
\item For each $i \in S\setminus S_j$, set $A_S(i,j)=0$ and $B_S(i,j)=1$.
\end{itemize}

\begin{theorem}{\cite{DP}}\label{thm:DP}
\begin{itemize}
\item[$(a)$] For any given $\eps < 1$, if $(\xx,\yy)$ is an $\eps$-\WNE of game $(A_S,B_S)$, then $(i)$ for all $i \notin S$, $x_i=0$, and $(ii)$ $||\xx - \ud{n,S}||_1 \le 8 \eps$.
\item[$(b)$] For $V=\{\ud{n,S} \ |\ S\subseteq[n], |S|=l\}$ we have $|V|= \Omega(n^{0.8\log_2 n})$. And for any $\gamma>0$ there exists $V'\subset V$ of size $\Omega(n^{(0.8-2\gamma\eps)\log_2 n})$ such that for any two vectors $\ud{n,S}, \ud{n,S'}\in V'$ we have $||\ud{n,S'} - \ud{n,S}||_1 > \gamma\eps$.
\end{itemize}
\end{theorem}

Note that, column player playing each of its $n$ strategies uniformly at random is her NE strategy in game $(A_S,B_S)$ for all $S\subset [n]$ of size $l$. Thus querying Oracle \ref{or2} with strategy $\yy=\ud{n}$ for the column player will return a NE of $(A_S,B_S)$. To prevent this next we extend this construction that uses the above game combined with another copy of it where roles of the two players are switched. 
For the construction to work as desired, we need both players to put non-trivial probability mass on both games, and to achieve this we augment this game with a matching-pennies game type construction.

\paragraph{Our Construction}
Again fix an even integer $l$ and let $n={l \choose l/2}$. Fix an $l$ sized subset $S\subset [n]$, and corresponding game $(A_S, B_S)$ as described above. We construct a $2n \times 2n$ game $(R_S, C_S)$ as follows: if we think of matrices of this game as $2 \times 2$ block matrix, each block of size $n\times n$, then

\begin{equation}\label{eq:const-goodgame}
R_S = \left[\begin{array}{cc}2+A_S & -2\\
-2 & 2+B_S^{\top}\end{array}\right]
\ \ \ \ \ 
C_S = \left[\begin{array}{cc}-2+B_S & 2 \\
2 & -2+A_S^{\top} \end{array}\right]
\end{equation}

In \eqref{eq:const-goodgame}, the constants $2$ and $-2$ represents block matrices of size $n\times n$ with all entries set to $2$ and $-2$ respectively. 
For any strategy profile $(\xx,\yy)$ of game $(R_S,C_S)$, let us divide each of $\xx$ and $\yy$ into two vectors corresponding to first $n$ and last $n$ strategies. These are $\xx^L = (x_1,\dots,x_n)$, $\xx^R = (x_{(n+1)},\dots,x_{2n})$, $\yy^L=(y_1,\dots,y_n)$, and $\yy^R=(y_{(n+1)},\dots,y_{2n})$. For any vector $\vv$, we will use $\sigma(\vv)$ to denote sum of its coordinates $(\sum_i v_i)$. 

Using the matching pennies property of game $(R_S,C_S)$, next we show that both $\xx^L$ and $\xx^R$ have enough probability mass, and similarly $\yy^L$ and $\yy^R$.
\begin{lemma}\label{lem:const-goodgame1}
For an $\eps\in [0,1)$, let $(\xx,\yy)$ be an $\eps$-\WNE of game $(R_S,C_S)$. Then, 
$\frac{9}{19} - \frac{\eps}{9} \le \sigma(\xx^L), \sigma(\xx^R), \sigma(\yy^L), \sigma(\yy^R) \le \frac{10}{19} +\frac{\eps}{9}$. 
\end{lemma}
\begin{proof}
First we show that $\sigma(\xx^L), \sigma(\xx^R), \sigma(\yy^L), \sigma(\yy^R) >0$. To the contrary suppose $\sigma(\xx^R)=0$, then column player gets payoff at most $-1$ from any stragety $j \le n$, while gets payoff of $2$ from any $k > n$. Therefore $\sigma(\yy^L)=0$. In that case, the row player gets payoff at most $-1$ from any $i \le n$, while gets at least $2$ payoff from any $h>n$. Since there exists $i \le n$ such that $x_i>0$, this contradicts $(\xx,\yy)$ being $\eps$-\WNE. Similar argument follows for each of $\sigma(\xx^L), \sigma(\yy^L), \sigma(\yy^R)$ being zero.

Next we will argue that $\frac{9}{19} - \frac{\eps}{9} \le \sigma(\yy^L) \le \frac{10}{19} +\frac{\eps}{9}$, and then using $\sigma(\yy^R)=1-\sigma(\yy^L)$ the inequalities will follow for $\sigma(\yy^R)$ as well. For this we need properties of payoffs from $A_S$ and $B_S$ blocks of $R_S$. 

\begin{claim}\label{cl:const1}
For all $k \le n$ we have $\et{k}A_S\yy^L \le 1$ and $\et{k}B_S^{\top}\yy^R \le 1$, and there exists an $i,i'\le n$ such that $\et{i}A_S\yy^L \ge \frac{\sigma(\yy^L)}{2}$ and $\et{i'}B_S^\top \yy^R\ge \frac{\sigma(\yy^R)}{2}$. 
\end{claim}
\begin{proof}
The first part follows by the fact that coordinates of both $A_S$ and $B_S$ are at most one. For the second part, note that for each column $j\in[n]$ of $A_S$ exactly $(l/2)$ coordinates are $1$ out of $l$ coordinates of $S$. Therefore, summing up payoffs corresponding to strategies of $S$, we have
\[
\sum_{i \in S} \et{i}A_S\yy^L = \sum_{j \le [n]} \yy^L *(l/2) = \sigma(\yy^L) (l/2)
\]
Since $|S|$ is $l$, among the strategies of $S$ there is at least one $i \in S$ such that $\et{i}A_S\yy^L \ge \frac{\sigma(\yy^L)}{2}$. 
For the second part, a row of $B_S^\top$ is a column of $B_S$. Note that every entry in $B_S$ is at least zero, and in a column all entries corresponding to rows in $S$ are one. Let us divide $\yy^R$ further into coordinates corresponding to set $S$, namely $\yy^{R,S}$, and those outside $S$, namely $\yy^{R,\bS}$. In block $B_S$ the role of first player is that of column-player in game $(A_S,B_S)$. Therefore, strategies $(n+1),\dots,2n$ correspond to the $n={l \choose (l/2)}$ subsets of $S$. There is one such set that correspond to largest $(l/2)$ coordinates of $\yy^{R,S}$. Let this be $S_{i'}$. Then, $\et{i'}B_S^\top\yy^R \ge \sigma(\yy^{R,\bS}) + \sigma(\yy^{R,S})/2 \ge \sigma(\yy^R)/2$. 
\end{proof}

Consider $i,i' \le n$ as in Claim \ref{cl:const1}. Since $\sigma(\xx^L),\sigma(\xx^R)>0$, there exists $h,h' \le n$ such that $x_h,x_{(n+h')} >0$. 
Row player's payoffs from strategies $i,(n+i'),h,$ and $(n+h')$ are:
\[
\begin{array}{l}
\et{h}R_S\yy = 2\sigma(\yy^L) + \et{h}A_S\yy^L - 2 \sigma(\yy^R), \ \ \ \et{(n+i')}R_S\yy = -2 \sigma(\yy^L) + \et{i'}B_S^\top\yy^R + 2 \sigma(\yy^R) \\
\et{i}R_S\yy= 2\sigma(\yy^L) + \et{i}A_S\yy^L - 2 \sigma(\yy^R), \ \ \ \et{(n+h')}R_S\yy = -2 \sigma(\yy^L) + \et{h'}B_S^\top\yy^R + 2 \sigma(\yy^R)
\end{array}
\]

Since $(\xx,\yy)$ is an $\eps$-\WNE we must have $\et{h}R_S\yy \ge \et{(n+i')}R_S\yy - \eps$ and $\et{(n+h')}R_S\yy \ge \et{i}R_S\yy -\eps$. Using Claim \ref{cl:const1}, the former gives,
\[
\begin{array}{lcl}
3\sigma(\yy^L) - 2\sigma(\yy^R) \ge -2\sigma(\yy^L) + \frac{\sigma(\yy^R)}{2}+ 2\sigma(\yy^R) -\eps & \Rightarrow & 5\sigma(\yy^L) \ge \frac{9}{2} \sigma(\yy^R) -\eps\\ 
& \Rightarrow & \sigma(\yy^L) \ge \frac{9}{10}\sigma(\yy^R) -\frac{\eps}{5} = \frac{9}{10}(1-\sigma(\yy^L)) -\frac{\eps}{5}\\
& \Rightarrow & (1+\frac{9}{10}) \sigma(\yy^L) \ge \frac{9}{10} -\frac{\eps}{5}\\
& \Rightarrow & \sigma(\yy^L) \ge \frac{9}{19} - \frac{2}{19}\eps\ge \frac{9}{19} - \frac{\eps}{9}
\end{array}
\]
Similarly, using $\et{(n+h')}R_S\yy\ge \et{i}R_S\yy -\eps$ together with Claim \ref{cl:const1} we will get $\sigma(\yy^R) \ge \frac{9}{10}\sigma(\yy^L) -\frac{\eps}{5}$. Simplifying this using $\sigma(\yy^R) = 1-\sigma(\yy^L)$ gives $\sigma(\yy^L) \le \frac{10}{19}+\frac{\eps}{9}$. 

Bounds on $\sigma(\xx^L)$ and $\sigma(\xx^R)$ can be obtained through similar arguments for the payoff of the column-player w.r.t. $C_S$. 
\end{proof}

For any $\eps<1$ the bounds obtained in Lemma \ref{lem:const-goodgame1} implies lower bound of at least $\frac{1}{3}$ and upper bound of at most $\frac{2}{3}$ on the each of $\sigma(\xx^L), \sigma(\xx^R), \sigma(\yy^L)$, and $\sigma(\yy^R)$.

\begin{lemma}\label{lem:const-goodgame2}
For an $\eps\in [0,1/3)$, let $(\xx,\yy)$ be an $\eps$-\WNE of game $(R_S,C_S)$. Then, $(\frac{\xx^L}{\sigma(\xx^L)}, \frac{\yy^L}{\sigma(\yy^L)})$ and $(\frac{\yy^R}{\sigma(\yy^R)}, \frac{\xx^R}{\sigma(\xx^R)})$ are $3\eps$-\WNE of game $(A_S,B_S)$. 
\end{lemma}
\begin{proof}
If $x_i$ for an $i\le n$, then it must be the case that $\et{i}R_S\yy \ge \et{k}R_S\yy -\eps$ for all $k\le n$. This implies,
\[
\et{i}A_S\yy^L \ge \et{k}A_S\yy^L - \eps \Rightarrow \frac{\et{i}A_S\yy^L}{\sigma(\yy^L)} \ge \frac{\et{k}A_S\yy^L}{\sigma(\yy^L)} - \frac{\eps}{\sigma(\yy^L)} \Rightarrow (\et{i}A_S\frac{\yy^L}{\sigma(\yy^L)}) \ge (\et{k}A_S\frac{\yy^L}{\sigma(\yy^L)}) - 3\eps
\]
The last inequality follows using Lemma \ref{lem:const-goodgame1}. Similarly, we can argue that if $y_j >0$ for a $j \le n$, then for all $k \le n$ we have $(\frac{\xx^{L^\top}}{\sigma(\xx^L)}B_S e_j) \ge (\frac{\xx^{L^\top}}{\sigma(\xx^L)}B_S e_k) - 3\eps$. These two together implies $(\frac{\xx^L}{\sigma(\xx^L)}, \frac{\yy^L}{\sigma(\yy^L)})$ is a $3\eps$-\WNE of game $(A_S,B_S)$.

In the lower diagonal block, the game is $(B_S^T, A_S^T)$, or in other words the roles of the two players are switched. Therefore, it follows that if $x_i>0$ for an $i>n$ then for all $k>n$ we have $(B_S^T\frac{\yy^R}{\sigma(\yy^R)} e_i) = (\frac{\yy^{R^\top}}{\sigma(\yy^R)}B_S e_i) \ge (\frac{\yy^{R^\top}}{\sigma(\yy^R)}B_S e_k) - 3\eps$. And similarly, if $y_j>0$ for $j>n$ then for all $k>n$ we have $(\et{j}A_S\frac{\xx^R}{\sigma(\xx^R)}) \ge (\et{k} A_S\frac{\xx^R}{\sigma(\xx^R)}) - 3\eps$. These two together implies $(\frac{\yy^R}{\sigma(\yy^R)}, \frac{\xx^R}{\sigma(\xx^R)})$ are $3\eps$-\WNE of game $(A_S,B_S)$.
The lemma follows.
\end{proof}

The above lemma together with Theorem \ref{thm:DP} implies that for any $\eps$-\WNE $(\xx,\yy)$, vector $\xx^L$ is near uniform distribution over strategies of $S$ scaled by $\sigma(\xx^L)$. In order to show non-overlapping equilibrium sets across games corresponding to different $S$, we will need that $\xx^L$ is near such a uniform distribution scaled by $\frac{1}{2}$. For this we need a tighter upper-lower bound on $\sigma(\xx^L)$, that we show next.  
The main missing component here is a tighter upper bound on the payoffs from $(A_S,B_S)$ proved in the next lemma.

\begin{lemma}\label{lem:const-goodgame3}
For an $\eps\in[0,1/3),$ if $(\xx,\yy)$ is an $\eps$-\WNE of game $(R_S,C_S)$, then $\forall i \le n,\ (\et{i}A_S\yy^L) \le \frac{\sigma(\yy^L)}{2} +4\eps$ and $(\xx^{L^\top}B_S e_i) \le \frac{\sigma(\xx^L)}{2} +4\eps$. And similarly, $(\et{i}B_S^\top\yy^R)\le \frac{\sigma(\yy^L)}{2} +4\eps$ and $(\et{i}A_S\xx^R) \le \frac{\sigma(\xx^R)}{2}+4\eps$ for each $i \in [n]$.
\end{lemma}
\begin{proof}
Let $\txx=\frac{\xx^L}{\sigma(\xx^L)}$ and $\tyy=\frac{\yy^L}{\sigma(\yy^L}$. Lemma \ref{lem:const-goodgame2} implies that $(\txx,\tyy)$ is a $3\eps$-\WNE of game $(A_S,B_S)$. Then due to Theorem \ref{thm:DP} part $(a)$, we have $\tx_i=0$ for all $i\notin S$. Since sum of the payoffs for rows in $S$ is always $1$ in this game, we have $\txx^\top\top (A_S+B_S)\tyy=1$. There are two possible scenarios:

{\bf Case I -} $\txx^\top A_S \tyy \le \frac{1}{2}$. 
Then it must be the case that $\forall i \le n,\ (\et{i}A_S\tyy) \le \frac{1}{2} + 3\eps$, or else $(\txx,\tyy)$ is not a $3\eps$-\WNE of $(A_S,B_S)$. 
Multiplying by $\sigma(\yy^L)$ gives, $(\et{i}A_S\yy^L) \le \frac{\sigma(\yy^L)}{2} + \sigma(\yy^L)3\eps \le \frac{\sigma(\yy^L)}{2} + 2\eps$ for each $i\in [n]$. The last inequality uses $\sigma(\yy^L) \le \frac{2}{3}$ from Lemma \ref{lem:const-goodgame1}. 

To get similar bound on $(\xx^{L^\top}B_Se_i)$, note that Claim \ref{cl:const1} gives existence of $i \in [n]$ such that $(\et{i}A_S\tyy) \ge \frac{1}{2}$. Therefore $\txx^\top A_S \tyy \ge \frac{1}{2} - 3\eps$, implying $\txx^\top B_S \tyy \le \frac{1}{2} + 3\eps$. Therefore, again $3\eps$-\WNE condition gives that for each $i\in [n]$ we should have $(\txx^\top B_Se_i) \le \frac{1}{2} + 6\eps$. Through similar analysis as above we would get that for each $i\in [n]$, $(\xx^{L^\top}B_Se_i) \le \frac{\sigma(\xx^L)}{2} + 4\eps$. 

{\bf Case II -} $\txx^\top B_S \tyy \le \frac{1}{2}$. 
Again by similar analysis as Case I, we will get that for each $i \in [n]$ we have $(\et{i}A_S\yy^L) \le \frac{\sigma(\yy^L)} + 4\eps$ and $(\xx^{L^\top}B_Se_i) \le \frac{\sigma(\xx^L)}{2} + 2\eps$.

The second part follows similarly using the fact that $(\frac{\yy^R}{\sigma(\yy^R)},\frac{\xx^R}{\sigma(\xx^R)})$ is also a $3\eps$-\WNE of game $(A_S,B_S)$ from Lemma \ref{lem:const-goodgame2}.
\end{proof}

Now we are ready to prove a tighter lower/upper bound on $\sigma(\xx^L), \sigma(\xx^R),\sigma(\yy^L),\sigma(\yy^R).$

\begin{lemma}\label{lem:const-goodgame4}
Given an $\eps\in [0,1)$, if $(\xx,\yy)$ is an $\eps$-\WNE of game $(R_S,C_S)$ then $\frac{1}{2} - \eps \le \sigma(\xx^L), \sigma(\xx^R),\sigma(\yy^L),\sigma(\yy^R) \le \frac{1}{2} +\eps$. 
\end{lemma}
\begin{proof}
The proof follows similar argument as of the proof of Lemma \ref{lem:const-goodgame1}. 
We will first show the claim for $\sigma(\yy^L)$ and thereby it will follow for $\sigma(\yy^R)$ as well since $\yy^R=1-\yy^L$. Let $i,i'\in [n]$ be the indices from Claim \ref{cl:const1}. Let $h,h'\in [n]$ be such that $x_h,x_{(n+h')}>0$. Then, $(\et{h}R_S\yy) \ge (\et{(n+i')}R_S\yy) -\eps$ gives
\[
2\sigma(\yy^L) + (\et{i}A_S\yy^L) - 2\sigma(\yy^R) \ge -2\sigma(\yy^L) +(\et{i'}B_S^\top\yy^R)+ 2\sigma(\yy^R) -\eps \Rightarrow 4\sigma(\yy^L) + \frac{\sigma(\yy^L)}{2} +\eps \ge 2\sigma(\yy^R) + \frac{\sigma(\yy^R)}{2} -\eps
\]
The last inequality uses $(\et{i}A_S\yy^L) \le \frac{\sigma(\yy^L)}{2} + 4\eps$ from Lemma \ref{lem:const-goodgame3}. 
Simplifying the above gives $\sigma(\yy^L) \ge \sigma(\yy^R) -2\eps \Rightarrow \sigma(\yy^L) \ge \frac{1}{2} - \eps$. Similarly, simplifying $(\et{(n+h')}R_S\yy)\ge (\et{i}R_S\yy) -\eps$ by making use of $(\et{i}A_S\yy^L)\ge \frac{\sigma(\yy^L)}{2}$ from Claim \ref{cl:const1} and $(B_S^\top \yy^R)\le \frac{\sigma(\yy^R)}{2} +4\eps$ from Lemma \ref{lem:const-goodgame3}, we get $\sigma(\yy^L) \le \frac{1}{2} +\eps$. 

Bounds for $\sigma(\xx^L),\sigma(\xx^R)$ follows using similar argument for the payoffs of the column player. 
\end{proof}

Now we are ready to show that vector $\xx^L$ and $\yy^R$ corresponding to any $\eps$-\WNE are near uniform distribution on set $S$ scaled by $1/2$. 

\begin{lemma}\label{lem:const-goodgame5}
For an $\eps\in[0,1/3),$ if $(\xx,\yy)$ is an $\eps$-\WNE of game $(R_S,C_S)$, then $(i)$ $\supp(\xx^L), \supp(\yy^R) \subseteq S$, and $(ii)$ $||\xx^L - \ud{n,S}/2||_1 \le 17\eps$ and $||\yy^R - \ud{n,S}/2||_1 \le 17\eps$. 
\end{lemma}
\begin{proof}
From Lemma \ref{lem:const-goodgame2} we have that $(\xx^L/\sigma(\xx^L), \yy^L/\sigma(\yy^L))$ is a $3\eps$-\WNE of game $(A_S,B_S)$. Then part $(a)$ of Theorem \ref{thm:DP} implies $\supp(\xx^L) \subset S$ and $||\xx^L/\sigma(\xx^L) - \ud{n,S}||_1 \le 24\eps$. Multiplying by $\sigma(\xx^L)$ gives $||\xx^L - \sigma(\xx^L)\ud{n,S}||_1 \le \sigma(\xx^L) 24\eps$. Using the fact that $\sigma(\xx^L) \ge \frac{1}{2} -\eps$ from Lemma \ref{lem:const-goodgame4}, we get $||\xx^L - \frac{1}{2}\ud{n,S}||_1 -\eps \le 16\eps \Rightarrow ||\xx^L - \frac{1}{2}\ud{n,S}||_1 \le 17\eps$.

Similarly, using $(\yy^R/\sigma(\yy^R), \xx^R/\sigma(\xx^R))$ being a $3\eps$-\WNE of game $(A_S,B_S)$ from Lemma \ref{lem:const-goodgame2}, together with part $(a)$ of Theorem \ref{thm:DP} and lower bound $(1/2-\eps)$ on $\sigma(\yy^R)$ from Lemma \ref{lem:const-goodgame4}, we can show that $\supp(\yy^R)\subseteq S$ and $||\yy^R -\frac{1}{2}\ud{n,S}||_1 \le 17\eps$.
\end{proof}

Recall that our final goal is to construct a family of \good games for appropriate parameters of $\eps,\tau$ and $c$ such that their $\eps$-\WNE do not intersect for a constant $\eps$. Note that all payoff entries of an \good game (Definition \ref{def:goodgame}) are in $[-1,1]$, and the first requirement is that no matter what the opponent plays, a player always has a strategy that gives at least half payoff. In order to ensure these we need to scale game $(R_S,C_S)$ additive and multiplicatively by a constant. The third part requires renaming of the strategies of the column player so that vectors $\yy^L$ and $\yy^R$ get swapped. 


\begin{equation}\label{eq:const-goodgame2}
\begin{array}{ll}
\forall i \in [2n],\ \ \ \forall j \in[n],\ \ \ & \tR_S(i,j)= (R_S(i,n+j) + 3)/6, \ \ \ \tR_S(i,n+j)=(R_S(i,j)+3)/6, \\ 
& \tC_S(i,j) = (C_S(i,n+j)+3)/6, \ \ \ \tC_S(i,n+j) = (C_S(i,j)+3)/6
\end{array}
\end{equation}

\begin{lemma}\label{lem:const-goodgame-main}
For an $\eps \in [0,1/26)$, $\tau=102$ and $c=1/2$, game $(\tR_S,\tC_S)$ is an $(\eps,\tau,c)$-\good game. 
\end{lemma}
\begin{proof}
By construction, every entry in $(A_S,B_S)$ is either $0,1$ or $-1$, and thereby from (\ref{eq:const-goodgame}) every entry in $(R_S,C_S)$ is in $[-3,3]$. Thus, $\tR_S, \tC_S \in [-1,1]$. Note that we have swapped the first $n$ and last $n$ strategies of the column player. This is just renaming of the strategies. 

For property $(1)$ of the \good game, consider $i$ and $i'$ of Claim \ref{cl:const1}. Note that the proof of this claim does not use $\eps$-\WNE property of profile $(\xx,\yy)$. For any given profile $\yy$ of the column-player, the payoff of the row player from strategy $i$ is:
\[
(\et{i}\tR_S\yy) = 1/6 \sigma(\yy^L) + 5/6 \sigma(\yy^R) + 1/12 \sigma(\yy^R) \ge \frac{13}{12} (1/2 - \eps) \ge 1/2
\]
The last inequality follows using $\eps < 1/26$. Similarly, we can show that no matter what the row-player plays, column player has a strategy that gives at least $\frac{1}{2}$ payoff. 

For Property $(2)$ let $(\xx,\yy)$ be an $\eps$-\WNE of game $(\tR_S,\tC_S)$. Then, for $\yy'=(\yy^R, \yy^L)$, $(\xx,\yy)$ is a $6\eps$-\WNE of game $(R_S,C_S)$. Then, part $(i)$ of Lemma \ref{lem:const-goodgame5} ensures $\Supp(\xx^L),\Supp(\yy^L) \subset S$ implying Property $(2.a)$ of \good game. While the second part ensures $||\xx^L - \frac{1}{2}\ud{n,S}||_1, ||\yy^L - \frac{1}{2}\ud{n,S}||_1 \le 17 (6\eps) = 102\eps$ implying Property $(2.b)$. 
\end{proof}

\paragraph{Constructing a family}
Finally, to apply Theorem \ref{thm:absFamily}, its hypothesis $(h_2)$ requires a family of \good games parameterized by set $S$ with certain properties. Using the \good game of the previous section next we will construct such a family with $n^{O(\log(n))}$ many games. 
The next lemma follows simply using part $(2)$ of Theorem \ref{thm:DP}.

\begin{lemma}\label{lem:const-goodfamily}
Let $\CS=\{S\subset[n]\ |\ |S|=l\}$ and $\eps < 10^{-5}$. There exists $(n^{0.75\log(n)})$ sized subset $\CS'$ of $\CS$ such that for any $S,S'\in \CS'$, $S\neq S'$, we have $||\ud{n,S} - \ud{n,S'}||_1 > \frac{18}{c}(\tau+1)\eps$, where $\tau=102$ and $c=1/2$. 
\end{lemma}

Now we have all ingredients now to prove Lemma \ref{lem:enum-hard-const-approximate}. 

\begin{proof}[Proof of Lemma \ref{lem:enum-hard-const-approximate}]
Set $\eps < 10^{-5}$, $\tau=102$ and $c=1/2$ like in Lemma \ref{lem:const-goodfamily}. The lemma implies existence of a family of subsets $\CS'$ of set $[n]$ such that $||\ud{n,S} - \ud{n,S'}||_1 > \frac{18}{c}(\tau+1)\eps$. For every set $S\in \CS'$ we know that $|S|=l$. Therefore, game $(\tR_S, \tC_S)$ of \eqref{eq:const-goodgame2} where $(R_S,C_S)$ is from \eqref{eq:const-goodgame}, is an $(\eps,\tau,c)$-\good due to Lemma \ref{lem:const-goodgame-main}. By construction, we have $(\tR_S,\tC_S) \in [-1,1]^{n\times n}$.
\end{proof}

\addreferencesection
\bibliographystyle{amsalpha}{}
\bibliography{bib/mathreview,bib/dblp,bib/scholar,bib/custom,bib/misc}

\newpage
\appendix

\section{Reductions Within Sum-of-Squares}


\begin{proof}[Proof of Fact~\ref{fact:reductions-within-SoS}]
Let $\tilde{\mu}(x)$ be a degree $d$ pseudo-distribution satisfing the constraints of $\Psi_1(x)$. Let $\tilde{\nu}(y)$ be a signed measure on $\R^{n_2}$ defined by the associated expectation functional: $\pE_{\tilde{\nu}}[q(y)] = \pE_{\tilde{\mu}}[q(f_1(x),f_2(x), \ldots, f_{n_2}(x))$ for every polynomial $q$. Observe that if the degree of $q$ is $a$, then the degree of $q(y(x))$ for $y(x) = (f_1(x), f_2(x), \ldots, f_{n_2}(x))$, is $at$.

We claim that $\tilde{\nu}$ is a degree $\lfloor d/t \rfloor$-pseudo-distribution satisfying the constraints of $\Psi_2.$

First, observe that since $\pE_{\tilde{\nu}}$ is a linear operator, $\pE_{\tilde{\nu}}$ is a well-defined linear operator. Let us now verify the remaining properties.

First, consider any polynomial of the form  $q = \sum_{i} r_i g_i$ where $r_i$ are arbitrary polynomials and $deg(r_i) + deg(g_i)$ is at most $\lfloor d/t \rfloor$ in $y$. Then, $q(y(x))$ has degree at most $d$ in $x$. Further, for every $j$, we know that $g_j(y(x)) = 0$. Thus, $\pE_{\tilde{\nu}(y)} \sum_i r_i g_i(y) = \pE_{\tilde{\mu}(x)}  \sum_i r_i g_i (y(x)) = 0$.

Next, consider any polynomial $q = \sum_i r_i g_i + \sum_i s_i h_i$ for arbitrary polynomials $r_i$ and sum-of-squares polynomials $s_i$ such that each term appearing in the expansion has syntactic degree of at most $\lfloor d/t \rfloor$. 

We have: $\pE_{\tilde{\nu}(y)} \sum_i r_i g_i + \sum_i s_i h_i = \pE_{\tmu(x)} \sum_i r_i (y(x)) g_i(y(x)) + \sum_i s_i(y(x)) h_i(y(x))$. 

Now, $g_i(y(x)) = 0$ for every $i$. Further, $h_i(y(x)) = Q_i(y(x))$ for some sum-of-squares polynomial $Q$ in $y$. This also means that $Q_i(y(x))$ is a SoS polynomial in $x$. Finally, since $s_i$ is a SoS polynomial in $y$, $s_i(y(x))$ is a SoS polynomial in $x$. Finally, product of SoS polynomials is SoS, thus, $s_i (y(x)) h_i(y(x))$ must also be SoS.

Now, using that $\pE_{\tilde{\mu}(x)}$ is a valid pseudo-expectation of degree $d$ completes the proof.

\end{proof}





\end{document}